\documentclass[11pt,letterpaper]{article}
\usepackage[utf8x]{inputenc}
\usepackage{float}
\usepackage{thmtools} 
\usepackage{thm-restate}
\usepackage{caption}

\usepackage[linesnumbered,boxed]{algorithm2e}
\SetAlCapSkip{1em}
\SetAlCapNameFnt{\small}
\SetAlCapFnt{\small}

\usepackage{mathtools}
\usepackage{soul} 
\usepackage{xcolor} 

\usepackage{xspace,enumerate}

\usepackage[T1]{fontenc}
\usepackage[full]{textcomp}

\usepackage{mathtools}
\setul{0.5ex}{0.3ex} 
\setulcolor{red} 

\usepackage{amsmath}
\usepackage{tikz}
\usepackage{hyperref}
\usepackage{amsthm}
\usepackage{cleveref}
\usepackage{mdframed}
\usetikzlibrary{arrows.meta, positioning}
\newtheorem{theorem}{Theorem}[section]
\newtheorem*{theorem*}{Theorem}

\newtheorem*{proposition*}{Proposition}
\newtheorem{lemma}[theorem]{Lemma}
\newtheorem*{lemma*}{Lemma}
\newtheorem{corollary}[theorem]{Corollary}
\newtheorem*{conjecture*}{Conjecture}

\newtheorem*{fact*}{Fact}

\newtheorem*{hypothesis*}{Hypothesis}
\newtheorem{conjecture}[theorem]{Conjecture}
\newtheorem{claim}[theorem]{Claim}
\newtheorem*{claim*}{Claim}

\theoremstyle{definition}
\newtheorem{definition}[theorem]{Definition}

\newtheorem{protocol}[theorem]{Protocol}

\newtheorem{remark}[theorem]{Remark}
\newtheorem*{remark*}{Remark}

\newtheorem*{observation*}{Observation}

\usepackage[
letterpaper,
top=1in,
bottom=1in,
left=1in,
right=1in]{geometry}

\usepackage[scr=rsfso]{mathalfa}
\usepackage{bm} 
\let\mathbb\varmathbb

\usepackage{hyperref}
\hypersetup{
colorlinks=true,
urlcolor=blue,
linkcolor=blue,
citecolor=OliveGreen,
}
\usepackage{prettyref}

\newcommand{\Esymb}{\mathbb{E}}
\newcommand{\Psymb}{\mathbb{P}}

\DeclareMathOperator*{\E}{\Esymb}

\DeclareMathOperator*{\ProbOp}{\Psymb}

\renewcommand{\Pr}{\ProbOp}

\DeclareMathOperator{\Tr}{Tr}
\DeclareMathOperator{\poly}{poly}
\DeclareMathOperator{\negl}{negl}

\DeclareMathOperator{\polylog}{polylog}

\newcommand{\N}{\mathbb N}

\newcommand{\CS}{\mathcal{S}}
\newcommand{\CN}{\mathcal{N}}
\newcommand{\braket}[1]{\langle #1 \rangle}
\newcommand{\BE}{\mathbb{E}}

\newcommand{\cN}{\mathcal N}
\newcommand{\cO}{\mathcal O}

\newcommand{\eps}{\varepsilon}

\numberwithin{equation}{section}

\usepackage{relsize}
\usepackage[font=footnotesize]{caption}

\usepackage{appendix}
\usepackage{amssymb}
\usepackage{amsbsy}
\usepackage{amsfonts}
\usepackage{latexsym}
\usepackage{amsmath}
\usepackage{stmaryrd}

\usepackage{qcircuit}

\DeclareMathAlphabet{\pazocal}{OMS}{zplm}{m}{n}

\definecolor{BrickRed}{rgb}{0.8, 0.25, 0.33}

\definecolor{DarkBlue}{rgb}{0,0,0.8}  
\definecolor{DarkOrange}{rgb}{0.8,0.4,0}  
\def\mylinkcolor{DarkBlue}

\hypersetup{
	colorlinks=true, urlcolor=\mylinkcolor,
  linkcolor=\mylinkcolor, citecolor=\mylinkcolor}

\newcommand{\Id}{\mathbb{I}}

\newcommand{\ket}[1]{|#1\rangle}
\newcommand{\bra}[1]{\langle#1|}

\newcommand{\ketbra}[2]{|#1\rangle\! \langle #2|}
\newcommand{\reg}[1]{\mathrm{#1}}

\newcommand{\Gen}{\mathrm{Gen}}
\newcommand{\Ver}{\mathrm{Ver}}

\usepackage{authblk}

\usepackage[most]{tcolorbox}

\usepackage[most]{tcolorbox}
\tcbuselibrary{skins, breakable}
\newtcolorbox{longfbox}[1][]{%
  enhanced, 
  colframe=black, 
  colback=white, 
  boxrule=0.5mm, 
  title=#1 
}

\begin{document}

\title{The Hardness of Learning Quantum Circuits \\ and its Cryptographic Applications}

\author[1]{Bill Fefferman\footnote{wjf@uchicago.edu}}
\author[1]{Soumik Ghosh\footnote{soumikghosh@uchicago.edu}}
\author[2]{Makrand Sinha\footnote{msinha@illinois.edu}}
\author[3]{Henry Yuen\footnote{hyuen@cs.columbia.edu}}

\affil[1]{University of Chicago}
\affil[2]{University of Illinois Urbana-Champaign}
\affil[3]{Columbia University}

\date{}
\maketitle

\begin{abstract}
We show that concrete hardness assumptions about learning or cloning the output state of a random quantum circuit can be used as the foundation for secure quantum cryptography. In particular, under these assumptions we  construct secure one-way state generators (OWSGs), digital signature schemes, quantum bit commitments, and private key encryption schemes. We also discuss evidence for these hardness assumptions by analyzing the best-known quantum learning algorithms, as well as proving black-box lower bounds for cloning and learning given state preparation oracles.


Our random circuit-based constructions provide concrete instantiations of quantum cryptographic primitives whose security do not depend on the existence of one-way functions. The use of random circuits in our constructions also opens the door to {NISQ-friendly quantum cryptography}. We discuss noise tolerant versions of our OWSG and digital signature constructions which can potentially be implementable on noisy quantum computers connected by a quantum network. On the other hand, they are still secure against {noiseless} quantum adversaries, raising the intriguing possibility of a useful implementation of an end-to-end cryptographic protocol on near-term quantum computers. Finally, our explorations suggest that the rich interconnections between learning theory and cryptography in classical theoretical computer science also extend to the quantum setting.






\end{abstract}
 
\newpage


\section{Introduction}
\label{sec: intro}
\noindent The hardness of learning and modern cryptography are inextricably linked in the world of classical computing. On one hand, hard cryptographic problems have served as the basis of showing that various learning tasks are intractable (see e.g., \cite{kearns1994cryptographic, klivans2006cryptographic, daniely2014averagecasecomplexityimproper}). Conversely, the hardness of various learning tasks have been used to construct fundamental cryptographic primitives, like pseudorandom functions~\cite{IL90, blum_furst_kearns_lipton_1993, oliveira2016conspiracieslearningalgorithmscircuit}, and practical cryptographic protocols ~\cite{regev2024latticeslearningerrorsrandom, lyubashevsky_peikert_regev_2010}. Furthermore, this connection between learning and cryptography extends further and has shed light on fundamental questions in the related areas of pseudorandomness~\cite{nisan_wigderson_1994} and circuit lower bounds~\cite{linial_mansour_nisan_1993, carmosino_impagliazzo_kabanets_kolokolova_2016}.

In the quantum world, our understanding of the analogous interconnections is quite lacking. In one direction, some prior works have used classical cryptographic assumptions, like quantum-secure one-way functions or \textrm{LWE}, to argue about the hardness of learning quantum states and circuits in different contexts (see \cite{arunachalam2020quantum, zhao2024learning}), but the complexity of many fundamental quantum learning tasks remain open. However, the \emph{converse} direction --- using hardness of quantum learning as a foundation for cryptography --- has not received much attention, unlike the classical case\footnote{As of the last several months, this is starting to change; see \Cref{sec:related-work} for discussion of recent work on this topic.}. Arguably, one difficulty in establishing such a connection is that classical cryptographic primitives appear insufficient to fully capture the inherently quantum nature of tasks involving learning quantum states or unitaries~\cite{NLF24}. This raises an interesting question that forms the basis of this paper:
\begin{center}
\emph{Can we base quantum cryptography on the hardness of quantum learning?}
\end{center}

\subsection*{Our results}

We propose several fine-grained assumptions about quantum learning tasks related to random quantum circuits, give evidence for it in the black-box model, and use these assumptions to construct several quantum cryptographic primitives. Incidentally, this exploration also leads to several interesting results such as ``NISQ-friendly'' quantum cryptography and a deeper understanding of various quantum learning tasks and cryptographic primitives. In more detail:





\begin{itemize}

\item \textbf{Quantum hardness of learning assumptions.} We posit that it is computationally intractable to learn the output states of \emph{unknown} random quantum circuits (Computational No-Learning Assumption) or to clone them (Computational No-Cloning Assumption). We formulate these conjectures and discuss the differences between them in \Cref{section: intro_learning assumptions}. We also give evidence for these conjectures by proving lower bounds in the black-box model.

\item \textbf{Cryptography from hardness of quantum learning.} We show that these hardness of learning  assumptions can be used to construct quantum cryptographic primitives. In particular, we construct  one way state generators (\textrm{OWSG}s) and quantum-secure digital signatures based on the Computational No-Learning Assumption, and we construct quantum bit commitments based on the Computational No-Cloning Assumption. Recent results have uncovered the tantalizing possibility of reducing the hardness assumptions behind such primitives beyond what is possible in classical cryptography~\cite{ananth2022cryptographypseudorandomquantumstates, morimae2022quantum, brakerski2023computational, kretschmer2023quantumcryptographyalgorithmica,kretschmer2024quantumcomputableonewayfunctionsoneway}. 
Our hardness assumptions about random circuits do not appear to imply the existence of classical one-way functions, and thus our  constructions concretely instantiate this possibility.





\item \textbf{Quantum cryptography for the NISQ era.} The use of random circuits in our constructions also opens the door to \emph{NISQ-friendly quantum cryptography}. We demonstrate that our OWSG and digital signature constructions can be made noise tolerant and thus implementable on a noisy quantum computer, provided we have a quantum channel to pass quantum states as cryptographic keys. On the other hand, they are still secure against \emph{noiseless} quantum adversaries. This raises the intriguing possibility of a near-term demonstration of \emph{quantum cryptographic advantage}: using a NISQ device to run a cryptographic protocol whose security relies on fewer assumptions than what is possible in classical cryptography.

\end{itemize}

\subsection*{Organization} In \Cref{section: intro_learning assumptions}, we introduce our hardness of learning assumptions. \Cref{sec:intro-crypto} gives a summary of the constructions of various cryptographic primitives. \Cref{sec:disc} provides a discussion of other implications of our results. \Cref{sec:related-work} and  \Cref{sec:is-this-the-future?} briefly discuss recent works that are related to this paper as well as mention some interesting directions that arise for future explorations. \Cref{sec:prelims} introduces the notation and some preliminaries. The black-box lower bound to support our hardness assumptions are given in \Cref{sec:blackbox}. \Cref{sec:crypto} includes details of the cryptographic constructions.

\subsection{Our hardness of learning assumptions}
\label{section: intro_learning assumptions}
We propose two main assumptions regarding the hardness of quantum learning: the  Computational No-Learning Assumption and the Computational No-Cloning Assumption. We precisely define these conjectures and discuss evidence for them.

\subsubsection{The Computational No-Learning Assumption}

In what follows, we let $\mathscr{C}$ denote a class of quantum circuits; for example, a concrete choice of $\mathscr{C}$ is the ensemble of 1D brickwork circuits with depth $d = \log^2 n$.\footnote{We pick out $d = \log^2 n$ here because it happens to be a depth at which the existing state-of-the-art learning algorithms, like those in \cite{Huang_2023, landau2024learning, kim2024learningstatepreparationcircuits}, 
fail to be efficient. Even if the learning algorithms were to be improved, any depth beyond which the efficiency of those learning algorithms fail is a good design choice for our assumptions and cryptographic constructions.}

\begin{conjecture}[Computational No-Learning]
\label{conj:learning}
Let $\mathscr{C}$ be a class of circuits. The $(\mathscr{C},\eps,\delta)$-Computational No-Learning Assumption stipulates the following. For all polynomial-time quantum algorithms $A$, for all sufficiently large $n$, we have that
\[
    \Pr \left [ \begin{array}{c} D \in \mathscr{C}_n \\ |\langle C | D \rangle|^2 \geq \eps \end{array} : \begin{array}{c} C \leftarrow \mathscr{C}_n \\ D \leftarrow A(\ket{C}^{\otimes \poly(n)}) \end{array}  \right] \leq \delta 
\]
Here, $C \leftarrow \mathscr{C}_n$ denotes that $C$ is sampled uniformly from the subset of circuits in $\mathscr{C}$ that act on $n$ qubits, and $D$ is sampled from the output of the algorithm $A$ given polynomially-many copies of $\ket{C} = C \ket{0^n}$. The output $D$ is interpreted as a description of an $n$-qubit quantum circuit, and $\ket{D} = D \ket{0^n}$. 
\end{conjecture}

In words, the Computational No-Learning Assumption says that it a polynomial-time quantum algorithm $A$, given only polynomially-many copies of the output $\ket{C}$ of a randomly-chosen circuit $C$ from the circuit ensemble $\mathscr{C}$, cannot produce a classical description of another circuit $D$ from the same class $\mathscr{C}_n$ and whose output state $\ket{D}$ approximates $\ket{C}$. We note that the algorithm $A$ only has access to copies of the state\footnote{We stress that $\ket{C}$ does not mean the classical description of $C$, but rather the state generated by the circuit $C$ on the all zero input.} generated by the random circuit and the circuit description itself is \emph{not} known to the algorithm. Furthermore, \Cref{conj:learning} actually refers to an entire \emph{family} of assumptions, one for each choice of circuit class $\mathscr{C}$ and parameters $(\eps,\delta)$. Throughout this paper we consider a fixed circuit family for convenience (e.g., 1D brickwork circuits with depth $\log^2 n$ unless otherwise specified). We elaborate on the different parameter settings for $(\eps,\delta)$ below. For clarity we abbreviate ``Computational No-Learning'' as ``No-Learning''. Oftentimes we will abbreviate ``$(\eps,\eps)$-No-Learning'' as ``$\eps$-No-Learning''.

\paragraph{Relationship between parameter values.} $(\eps,\delta)$-No-Learning implies $(\eps',\delta')$-No-Learning for all $\eps' \geq \eps$ and $\delta' \geq \delta$. This is argued by taking the contrapositive: if there exists an efficient algorithm $A$ that produced an $\eps'$-fidelity approximation $\ket{D}$ of $\ket{C}$ with probability at least $\delta'$, then $A$ \emph{also} produced (with the same $D$) a $\eps$-fidelity approximation with probability at least $\delta$.

A conservative assumption would be to conjecture $(1 - 1/\poly(n))$-No-Learning (i.e., it is hard for a polynomial-time algorithm to produce a high-fidelity approximation of a state with high probability). On the other hand, given our current understanding of quantum state learning techniques, it also seems plausible to conjecture $2^{-\Omega(n)}$-No-Learning (i.e., it is hard for a polynomial-time algorithm to produce even an \emph{exponentially bad} approximation with \emph{exponential small} probability). That said, we remark that  for the purposes of implementing the cryptographic constructions in this paper, the parameters we require are fairly mild. For instance, we merely need the $(1 - 1/\poly(n))$-No-Learning assumption for the cryptographic constructions that assume a noise-free circuit, and even for our NISQ-friendly constructions, the $\eps$-No-Learning assumption for any negligible\footnote{A function is negligible if it goes to zero faster than any inverse polynomial.} function $\eps(n)$ is sufficient.

\paragraph{Edge cases.} The hardness assumption cannot hold for $\eps \ll 2^{-n}$; this is because a quantum algorithm can output a uniformly random circuit $D$, and $\ket{D}$ will have fidelity at least $2^{-n}$ with $\ket{C}$ in expectation. Similarly, the hardness assumption cannot hold for $\delta \leq |\mathscr{C}_n|^{-1} = 2^{-\Omega(nd)}$, because the algorithm can always simply guess the description $C$ of the state $\ket{C}$.

\paragraph{Proper vs. improper learning.} We point out that the assumption is about the hardness of producing a circuit description $D$ that \emph{comes from the same class $\mathscr{C}$} as $C$. In the learning theory literature this is known as the setting of \emph{proper} learning, where the learner has to output a hypothesis from the same concept class as the true underlying concept. One can also consider the setting of \emph{improper} learning, where the output description $D$ can come from a more general class of circuits. For example, the circuit learning algorithms of~\cite{huang2024certifyingquantumstatessinglequbit,landau2024learning} are improper learners: given copies of output states of depth $d$ circuits, they produce descriptions of circuits with depth greater than $d$. One can consider hardness of learning conjectures in the improper setting as well; ours is the more \emph{conservative} assumption (because the learner has the additional constraint of outputting a circuit from $\mathscr{C}$). As will be clear later, even if an improper learner exists for our class of circuits, our main cryptographic primitive designed on the hardness of learning---a random circuit based one-way state generator outlined in \Cref{sec:OWSG}---will remain secure. Qualitatively, this is because the honest parties know the depth of the circuit. So, even if an adversary tries to cheat by using an improper learner, the honest parties can easily detect the mismatch in depth from the circuit description sent by the adversary.


\subsubsection{The Computational No-Cloning assumption}
\noindent In this section, we postulate a computational version of the famous No-Cloning principle of quantum mechanics. At a high level, it stipulates that given polynomially copies of the output state of a random quantum circuit, it is intractable to produce an additional copy of the state, or even an approximation of it.

\begin{conjecture}[Computational No-Cloning]
\label{conj:no_cloning}
Let $\mathscr{C}$ denote a class of circuits. The $(\mathscr{C},\eps,\delta)$-Computational No-Cloning Assumption stipulates the following. For all polynomial-time quantum algorithms $A$, for all polynomials $k(n)$, for all sufficiently large $n$, we have that
\[
    \Pr \left [ |\langle C |^{\otimes (k+1)} |\phi \rangle|^2 \geq \eps : \begin{array}{c} C \leftarrow \mathscr{C}_n \\ \ket{\phi} \leftarrow A(\ket{C}^{\otimes k}) \end{array}  \right] \leq \delta 
\]
Here, $C \leftarrow \mathscr{C}_n$ denotes that $C$ is sampled uniformly from the subset of circuits in $\mathscr{C}$ that act on $n$ qubits, and a $(k+1)n$ qubit state $\ket{\phi}$ is sampled from the output of the algorithm $A$, given polynomially-many copies of $\ket{C} = C \ket{0^n}$. 
\end{conjecture}

As before, the algorithm $A$ only has access to copies of the state generated by the random circuit and not the circuit description itself.
We also abbreviate ``Computational No-Cloning'' as ``No-Cloning'' and also abbreviate ``$(\eps,\eps)$-No-Cloning'' as ``$\eps$-No-Cloning''. 

\begin{corollary}
\label{corollary_learningvscloning}
    The $\eps$-No-Cloning Assumption (Conjecture \ref{conj:no_cloning}) implies the $\eps^c$-No-Learning Assumption (Conjecture \ref{conj:learning}) for some $0 < c < 1$.
\end{corollary}

\begin{proof}[(Proof Sketch)]
This follows from taking the  contrapositive. If the description of the circuit $C$ could be approximately learned with squared fidelity $\epsilon^c$ by some quantum polynomial-time algorithm $A$ (i.e, quantum state learning is easy), we can devise an efficient cloning algorithm $B$ that first runs $A$ in a coherent fashion on the $k$ copies to learn an approximation $D$, synthesizes an extra copy of $\ket{D}$ (which is supposed to represent the $(k+1)$st copy of $\ket{C}$), and then uncomputes $A$ to recover the original $k$ copies of $\ket{C}$. The fact that this strategy obtains $\eps$ fidelity with $\ket{C}^{\otimes (k+1)}$ follows from a calculation virtually identical to that of \Cref{clm:hiding}, which is part of the proof that the No-Cloning Assumption implies the existence of quantum commitments. This violates the $\eps$-No-Cloning Assumption. 
\end{proof}

\paragraph{Cloning is potentially easier than learning.} Corollary \ref{corollary_learningvscloning} means that learning is potentially an \emph{easier} task than learning. It could be that we can clone just one more copy of the state but we still cannot learn the state. 

There is some suggestive evidence of this fact from the work of Nehoran and Zhandry~\cite{nehoran2024computational}, who construct a quantum oracle with respect to which a collection of states is efficiently clonable, but it is not efficiently ``telegraphable,'' given only one quantum sample. Informally, telegraphing means getting a classical string out of one copy of a quantum state by ``deonconstructing'' it, from which one copy of the state can again be ``reconstructed.'' While the notion is qualitatively reminiscent of a learning task, this is formally different from the notion of learning in our paper.  






\subsection{Barriers and evidence for the hardness assumptions}

\subsubsection{The barriers}

What are the prospects of proving \Cref{conj:learning} or \Cref{conj:no_cloning} outright? First, we note that the conjectures are necessarily \emph{computational}, meaning that they are intrinsically about the limits of efficient quantum computation. This is because it \emph{is} possible for an exponential-time quantum algorithm to learn a circuit description with high probability, given only polynomially-many quantum samples (and thus also solve the cloning task with similar sample complexity). One method is to use the classical shadows protocol of~\cite{Huang_2020}; we describe this in more detail in \Cref{subsection: classical shadows}. 
We note that Morimae and Yamakawa already observed that there is always an exponential-time attack on one-way state generators~\cite{morimae2022one} via shadow tomography; this is essentially the same observation. 

Furthermore, we note that proving our hardness conjectures would have dramatic implications for \emph{classical} complexity theory. The classical shadows-based algorithm described in \Cref{subsection: classical shadows} can be implemented in polynomial time assuming the complexity inclusion $\mathrm{NP}^{\# \mathrm{P}} \subseteq \mathrm{BQP}$.\footnote{It was recently shown by Hiroka and Hsieh~\cite{hiroka2024computational} that efficient state learning is possible if $\mathrm{PP} \subseteq \mathrm{BQP}$, which represents an improved complexity upper bound.} Therefore \Cref{conj:learning} implies $\mathrm{P} \neq \mathrm{PSPACE}$. While this doesn't imply (say) $\mathrm{P} \neq \mathrm{NP}$, for all intents and purposes this would represent a similar breakthrough in mathematics and complexity theory. The longstanding difficulty of proving such complexity separations pose a barrier to proving our hardness assumptions.

\subsubsection{The evidence}
\label{subsubsection: the evidence}
We now discuss the evidence for our hardness assumptions. First we discuss the No-Learning assumption.



\paragraph{Examining best known learning algorithms.} Recently there has been progress on obtaining efficient algorithms for learning states of bounded circuit complexity~\cite{Huang_2023,fefferman2024anticoncentrationunitaryhaarmeasure, landau2024learning}. The algorithm of Fefferman, Ghosh, and Zhan~\cite{fefferman2024anticoncentrationunitaryhaarmeasure} requires the use of oracle access to the circuit itself, rather than only having copies of the output state $C \ket{0^n}$, which we do not consider in this paper. For algorithms that learn using copies of the output state alone, the current state-of-the-art is due to Landau and Liu~\cite{landau2024learning}, who show the following: 

\begin{theorem}[\cite{landau2024learning}]
\label{landauliu}
Fix an integer $k > 0$. There is an algorithm that, for all $\eps > 0,\delta > 0$, given copies of an unknown quantum state $\ket{C} = C \ket{0^n}$ for some depth-$d$ circuit acting on a $k$-dimensional lattice with two-qubit gates, outputs the description of a depth $(2k+1)d$ circuit $D$ such that $\| \ket{C} - \ket{D} \| \leq \eps$ with probability $1 - \delta$. Furthermore, the algorithm uses
\[
    \frac{n^4 \cdot 2^{O(d^k)}}{\eps^4} \log \frac{n}{\delta}
\]
copies of $\ket{C}$, and runs in time
\[
    \frac{n^4 \cdot 2^{O(d^k)}}{\eps^4} \log \frac{n}{\delta} + \Big( \frac{nd}{\eps} \Big)^{O(d^{k+1})}~.
\]
Here, the $O(\cdot)$ suppresses dependence on $k$, which we treat as a constant. 
\end{theorem}
For intuition, consider some parameter settings. Suppose $\eps, \delta$ are some fixed constants (like $0.001$). 
\begin{enumerate}
    \item When $d = O(1)$, then the sample complexity is $O(n^4 \log n)$, and the time complexity is $n^{O(d^k)}$. 
    \item When $d = O(\log^{1/k} n)$, then the sample complexity is still $\poly(n)$, but the time complexity becomes quasipolynomial $2^{\polylog n}$.
    \item When $d \gg \log n$, then both the sample and time complexity become superpolynomial.
\end{enumerate}

In more detail, the learning algorithms of both Huang, et al.~\cite{Huang_2023} and Landau and Liu~\cite{landau2024learning} revolve around the idea of learning so-called ``local inversions'' of the circuit $C$, which are small subcircuits $V$ that ``undo'' part of the overall circuit to be learned: $ (V^\dagger \otimes \Id) \ket{C} \approx \ket{0^k} \otimes \ket{\theta}$ for some $n-k$ qubit state $\ket{\theta}$. In other words, some small number $k$ of qubits have been disentangled from the state $\ket{\psi}$. 

If the circuit depth $d$ is small, then the local inversions have size $O(d^k)$ (assuming a $k$-dimensional circuit architecture), and the learning algorithm can brute force over all possible subcircuits of size $O(d^k)$. Once the local inversions have been learned, the challenge is to ``stitch'' all of the local inversions (which overlap with each other) in a consistent way. 

The complexity of learning a single local inversion takes time at least $2^{O(d^k)}$. Thus with a randomly chosen circuit of depth that asymptotically grows faster than $O(\log n)$ (e.g., $d = \log^2 n$), there are many more possibilities than is possible to consider in polynomial time. Furthermore, searching through candidate local inversions of a random circuit appears to be an ``all or nothing'' task: either a candidate successfully inverts a local patch, or it will likely scramble the state even further. Thus it does not seem possible to make the learning algorithms of~\cite{Huang_2023,landau2024learning} ``gracefully fail'' by dialing back their runtime to being polynomial in $n$. Looking at the prototypical algorithmic ideas in these papers suggests that if one is in fact limited to polynomial time algorithms, then the typical fidelity will likely be exponentially small. 

There are other papers, like \cite{kim2024learningstatepreparationcircuits}, that take as input copies of the local density matrix on $\mathcal{O}(1)$-sized patches as input and output the description of the global unitary. \cite{kim2024learningstatepreparationcircuits} uses an information-theoretic criteria and learns local quantum channels that can extend the state. However, this approach is also similar in spirit to the local inversion technique and the time complexity has a factor that scales as $2^{O(d^k)}$, which means that, similar to \cite{Huang_2023},  these algorithms are inefficient for any depth that grows asymptotically larger than $\mathcal{O}(\log n)$.

\paragraph{Worst-case hardness from post-quantum assumptions.} The No-Learning Assumption is an \emph{average-case} hardness assumption, because it is about learning over the uniform distribution over circuits. One can also wonder about the hardness of \emph{worst-case} learning. Zhao, et al.~\cite{zhao2024learning} showed that, assuming that subexponential-time quantum computers cannot solve the RingLWE problem, any quantum algorithm that learns the circuit description given copies of the $n$-qubit state requires at least $\exp(\Omega(\min\{G,n\}))$ time, where $G$ is size of the quantum circuit to be learned. RingLWE is a version of the Learning With Errors (LWE) problem, whose assumed hardness underlies most proposals for post-quantum cryptography (i.e., cryptography that can be run on classical computers, but are secure against quantum computers). When the number of gates $G$ is superlogarithmic, then the time complexity lower bound is superpolynomial. 

One may wonder why, given the results of Zhao, et al.~\cite{zhao2024learning}, we need to make a separate assumption about the hardness of learning quantum circuits. Given the widespread belief in the security of various versions of LWE (for example this underlies much of NIST's recent standardization of recommended post-quantum cryptosystems~\cite{fips2032023module,NIST_FIPS_204}), it would seem that hardness of (Ring)LWE automatically implies the hardness of learning circuits from~\cite{zhao2024learning} and therefore all our applications of \Cref{conj:learning} follow. 

There are two issues with this reasoning. The first is that the result of Zhao, et al.~\cite{zhao2024learning} implies the hardness of learning under a very particular distribution of quantum circuits: namely, ones that encode the RingLWE problem, which are quite structured\footnote{In more detail, these are based on constructions of \emph{pseudorandom states} from one-way functions~\cite{ji2018pseudorandom}.} and are statistically far from truly random quantum circuits. The second and most important issue is that, while the RingLWE hardness can be viewed as evidence for our hardness conjectures, it appears to be a significantly harder statement to prove. For one, the hardness of RingLWE trivially implies $\mathrm{P} \neq \mathrm{NP}$, one of the major open questions in mathematics. On the other hand, $\mathrm{P} \neq \mathrm{NP}$ is \emph{not} known to be implied by \Cref{conj:learning}. Indeed, there is emerging evidence that \Cref{conj:learning} is a reduced assumption compared to $\mathrm{P} \neq \mathrm{NP}$ (i.e., there are mathematical worlds in which $\mathrm{P} = \mathrm{NP}$ but tasks related to learning quantum circuits is still hard~\cite{kretschmer2023quantumcryptographyalgorithmica}). Since one of our primary motivations is to base cryptography on the fewest mathematical assumptions possible, we treat our hardness assumptions as being more basic and plausible than any post-quantum hardness assumption.

\paragraph{Evidence for No-Cloning assumption.} We can similarly examine the evidence for the No-Cloning assumption. As mentioned above, since the cloning is potentially an easier task, the No-Cloning conjecture is a \emph{stronger} assumption than No-Learning. However, as far as we are aware, there are no known algorithms for cloning that perform significantly better than ones for learning. One could consider, for example, the optimal pure-state cloning map that was analyzed by Werner~\cite{werner1998optimal}.  This map, which is essentially to project the input state onto the $(k+1)$-fold symmetric subspace, is provably the most sample-efficient procedure for taking copies $\ket{\psi}^{\otimes k}$ of an \emph{arbitrary} input state $\ket{\psi}$ (not necessarily one generated by a polynomial-size circuit) and producing an approximation of $\ket{\psi}^{\otimes k+1}$. Werner~\cite{werner1998optimal} showed the best achievable fidelity in general is
\[
    \frac{\binom{2^n + k - 1}{k}}{\binom{2^n + k}{k+1}} = \frac{k+1}{2^n}
\]
which is exponentially small for $k = \poly(n)$. 

Aside from Werner's optimal cloning map, as far as we are aware the best algorithms for the cloning task are  based on first solving the learning task. 
As we discussed before, such algorithmic ideas, if we restrict them to run in polynomial time, seem unable to achieve anything better than exponentially small fidelity. We leave it is an interesting open problem to come up with an algorithms -- even ones that run in subexponential time -- for cloning states of random circuits that do not learn the circuit description first.

\paragraph{Black-box lower bounds.} We give more evidence to support the No-Learning and No-Cloning assumptions by proving lower bounds in the black-box model. We desire a black-box model which captures the analogous properties in the white-box setting. For instance, one  property we would like to capture is the typical distribution of the output state of a random quantum circuit, which mimics the distribution of a Haar-random state. Another property we would like to capture in the black-box model is the existence of a learning algorithm that uses only polynomially-many copies of the state but is allowed to make exponentially many black-box queries -- this would be analogous to the existence of the exponential-time, polynomial-sample complexity learning algorithm.

With such considerations in mind, we formalize a \emph{state preparation oracle} model where the oracle $\mathcal{O}$ takes as input a state $\ket{i} \ket{0^n}$ and outputs $\ket{i} \ket{\psi_i}$ for some Haar-random state $\ket{\psi_i}$. The oracle can be accessed in superposition, and there is no guarantee about what the oracle does when the second register is initialized to something other than all zeroes. Intuitively, each index $i$ corresponds to a different polynomial-size circuit, but the algorithm is not allowed to exploit the structure of the circuit except to prepare the resulting output state. We prove the following:
\begin{restatable}[Black-box lower bounds for cloning]{theorem}{blackbox}
\label{thm:blackbox}
    There exists a state preparation oracle $\mathcal{O}$ such that all $T$-query quantum query algorithms getting $k$ copies of $\ket{\psi_J}$ for a uniformly random index $J \in [2^n]$ satisfy
    \[
        \mathrm{F}(\rho,\ketbra{\psi_J}{\psi_J}^{\otimes k+1}) \leq 2^{-n/4} ( 2T + k + 1)
    \]
    where $\rho$ is the output of the query algorithm and $\mathrm{F}(\cdot)$ denotes the fidelity function.
\end{restatable}
In other words, unless either the number of queries $T$ or the number of copies $k$ are $2^{\Omega(n)}$, the expected fidelity of cloning is exponentially small. Since the ability to learn implies the ability to clone, this also shows a similar lower bound for the learning task. We elaborate on the model and prove \Cref{thm:blackbox} in \Cref{sec:blackbox}.





\subsection{Cryptography from our hardness assumptions}
\label{sec:intro-crypto}

We now summarize our main cryptographic applications using the hardness of quantum learning assumptions from Section \ref{section: intro_learning assumptions}.

\subsubsection{One-way state generators from hardness of learning}
\label{sec:OWSG}
A \emph{one way state generator} (\textrm{OWSG}) is an efficient algorithm that takes as input a classical key $k$, and outputs a quantum state $\ket{\psi_k}$ from that key. Anyone with the key $k$ can efficiently verify that $\ket{\psi_k}$ is the correct output. Furthermore, given polynomially many copies of the $\ket{\psi_k}$ for a randomly chosen $k$, it should be computationally hard for an adversary to produce another key $k'$ such that the corresponding state $\ket{\psi_{k'}}$ is close to the original output $\ket{\psi_k}$. (For a formal definition, see \Cref{sec:OWSG}). Introduced by Morimae and Yamakawa~\cite{morimae2022quantum}, OWSGs are a quantum analogue of one-way functions, which are functions efficiently computable in the forwards direction but computationally difficult to invert.

The No-Learning Assumption (\Cref{conj:learning}) is essentially \emph{equivalent} to the existence of a OWSG, namely the Random Circuit OWSG described below in \Cref{fig:owsg}. For a range of parameters, this is immediate: for negligible $\eps$ (i.e., $\eps$ goes to $0$ faster than any inverse polynomial), the $\eps$-No-Learning Assumption is easily seen to be equivalent to the security of the Random Circuit OWSG. When $\eps$ is larger, say even up to $1 - 1/\poly(n)$, the equivalence still holds; this relies on hardness amplification techniques for OWSGs~\cite{morimae2022one,Bostanci_2024}. We note that this equivalence was also independently observed by Hiroka and Hsieh in a recent preprint~\cite{hiroka2024computational}.


\begin{figure}
\begin{longfbox}
    \begin{protocol} {\bf Random Circuit \textrm{OWSG}} \label{prot:OWSG_intro} 
    \end{protocol}

    \textbf{Generation algorithm}: Given input key $C \in \{0,1\}^{r(n)}$, interpret it as a description of an $n$-qubit circuit $C$ from the ensemble $\mathscr{C}_n$. Output $\ket{C} = C \ket{0^n}$. 

    \vspace{4pt}

    \textbf{Verification procedure}: Given input $C \in \{0,1\}^{r(n)}$ and a state $\ket{D}$ on $n$ qubits, apply $C^\dagger$ to the state, and measure. Accept if the result is all zeroes, and reject otherwise. 
\end{longfbox}
\caption{Construction of one-way state generator from random circuits.}
\label{fig:owsg}
\end{figure}

Although a OWSG is not immediately cryptographically useful by itself, it is now known that OWSGs can be used as a primitive building block for a variety of quantum cryptosystems. In their papers defining OWSGs~\cite{morimae2022quantum,morimae2022one}, Morimae and Yamakawa showed that OWSGs can be used to build bounded-time-secure digital signature schemes. In a breakthrough work, Khurana and Tomer showed that OWSGs imply the existence of quantum bit commitments~\cite{khurana2024commitments}, which in turn imply other functionalities such as quantum zero knowledge proofs for $\mathrm{NP}$ and secure multiparty computation~\cite{brakerski2023computational}. Therefore, using the Random Circuit OWSG, we obtain concrete implementations of these cryptosystems on quantum computers, without relying on the use of one-way functions.

An attractive feature of our Random Circuit OWSG is that it seems potentially amenable to implementation on noisy quantum computers, and thus the corresponding cryptosystems may be realizable in the near- and medium-term. We discuss noise tolerant versions of our random circuit-based cryptographic protocols in \Cref{sec:intro-NISQ}.

\subsubsection{Simple quantum commitments from the hardness of cloning} 

A commitment scheme enables two parties (known as a ``committer'' and a ``receiver'') to perform the cryptographic equivalent of putting a message in a sealed envelope that is opened later. 
In a quantum commitment scheme, a committer upon getting a bit $b$ generates a bipartite pure state $\ket{\psi_b}_{\reg{AB}}$, and sends the register $\reg{B}$ to the receiver; this is the ``commitment phase'' of the protocol and is the analogue of sending the sealed envelope. At this point the receiver should not be able to tell what the bit $b$ is.

Later, in the ``reveal phase'' of the protocol, the committer announces the bit $b$ and sends the remaining register $\reg{A}$ of $\ket{\psi_b}$ to the receiver, who can uncompute the state to check its validity. The security of the commitment scheme ensures that the committer cannot ``change his mind'' in between the commit and reveal phases to convince the receiver he had committed to the opposite bit $1 - b$. 

Recently, quantum commitments have become a centerpiece of the zoo of quantum cryptographic primitives~\cite{ananth2022cryptographypseudorandomquantumstates,morimae2022quantum,brakerski2023computational,khurana2024commitments}. As mentioned, Khurana and Tomer showed that OWSGs can be used to construct quantum bit commitments in a generic way~\cite{khurana2024commitments}. Therefore our No-Learning Assumption implies the existence of secure quantum commitments. However the construction of commitments in~\cite{khurana2024commitments} is quite involved, requiring an intricate sequence of transformations mimicking the classical transformation from one-way functions to pseudorandom generators~\cite{haastad1999pseudorandom}. 

We construct a simple quantum bit commitments based on the Computational No-Cloning Assumption (\Cref{conj:no_cloning}), and furthermore the analysis is fairly direct and straightforward. We describe the construction below. As discussed in Section \ref{section: intro_learning assumptions}, the No-Cloning Assumption implies the No-Learning Assumption. The ease of obtaining a commitment scheme from the No-Cloning Assumption suggests that the No-Cloning Assumption could be strictly stronger than the No-Learning Assumption (because it appears that obtaining commitments from OWSGs requires an intricate analysis~\cite{khurana2024commitments}).

\begin{figure}
\begin{longfbox}
    \begin{protocol} {\bf Commitment scheme based on random circuits} \label{prot:commitment_intro} 
    \end{protocol}

    \textbf{Commitment phase:} To commit to bit $b = 0$, the committer prepares the state 
    \[
        \ket{\psi_0}_{\reg{AB}} := \frac{1}{\sqrt{|\mathscr{C}_n|}} \sum_{C \in \mathscr{C}_n} \Big( \ket{C}^{\otimes k} \otimes \ket{0^n} \Big)_{\reg{A}} \otimes \ket{\hat{C} }_{\reg{B}}
    \] 
    where $\ket{\hat{C}}$ represents the \emph{classical description} of the circuit $C$. \\
    
    To commit to bit $b = 1$, then prepare the state
    \[
        \ket{\psi_1}_{\reg{AB}} := \frac{1}{\sqrt{|\mathscr{C}_n|}} \sum_{C \in \mathscr{C}_n}  \ket{C}^{\otimes (k+1)}_{\reg{A}} \otimes \ket{\hat{C} }_{\reg{B}}~.
    \] 

    The committer sends register $\reg{B}$ of the state $\ket{\psi_b}$ to the receiver.
    
    \vspace{4pt}
    
    \noindent \textbf{Reveal phase.} The committer reveals the bit $b$ to the receiver, and also sends the remaining register $\reg{A}$ of $\ket{\psi_b}$. To verify, the receiver will uncompute the unitary that synthesizes $\ket{\psi_b}$ and check that the all zeroes state is obtained.
\end{longfbox}
\label{fig:commitment-from-no-cloning}
\caption{Commitment scheme based on random circuits}
\end{figure}


We present and analyze the bit commitment scheme in detail in \Cref{bit commitments}.

\subsection{NISQ-friendly quantum cryptography}
\label{sec:intro-NISQ}

An important challenge in the field of quantum computing is to find a practical use case for noisy, near-term quantum (i.e., NISQ) computers. Although great strides have been made recently in demonstrating principles of error-correction and fault-tolerance on quantum devices~\cite{bluvstein2024logical,acharya2024quantum}, large-scale implementations of many quantum algorithms of interest (e.g., Shor's, Grover's, Hamiltonian simulation, etc) still seem quite a ways off. Thus there is significant interest in finding an application that (a) can be implemented on an a NISQ device, (b) has some advantage over classical computers/protocols, and (c) is practically useful. We show that our hardness assumptions yield quantum cryptosystems that satisfy these three criteria. 

Our random circuits-based cryptosystems are arguably NISQ-friendly. Random quantum circuits have been investigated extensively on real hardware platforms ever since Google's original quantum supremacy announcement in 2019~\cite{Arute2019}. The lack of structure in random quantum circuits is advantageous for maximizing quantum advantage while minimizing the burden on the NISQ device. Furthermore, the effective noise model when executing random quantum circuits often becomes quite simple~\cite{Arute2019,dalzell2024random}.

Furthermore, implementing cryptosystems based on our hardness assumptions would concretely realize the possibility of having secure quantum cryptography using reduced assumptions as compared to classical cryptography (for example, we do not need to assume $\mathrm{P} \neq \mathrm{NP}$ or that one-way functions exist)~\cite{kretschmer2,ananth2022cryptographypseudorandomquantumstates,morimae2022quantum,brakerski2023computational,khurana2024founding}. This would represent what we call ``quantum cryptographic advantage.''

    
    
We presented protocols that solve useful cryptographic tasks: digital signatures, bit commitments,  encryption, and more. Although these protocols are not immediately NISQ-friendly out of the box, we show how to ``NISQ-ify'' some of them so that they are. 

\subsubsection{NISQ-friendly one-way state generators}
 While implementing the Random Circuit \textrm{OWSG} on a realistic quantum computer, noise can degrade the fidelity of the output state. For certain choices of noise parameters and depth regimes, the fidelity can be inverse polynomial in the number of qubits. If that happens, the success probability of any verification procedure  would degrade similarly. Although the OWSG may be secure against any polynomial-time adversary, it may not be very useful if the ``honest user'' (e.g., someone with the key) tries to run it on a noisy quantum computer. On the other hand, making a OWSG tolerant to noise may give greater leeway to break its security. 


However, if we make a sufficiently strong assumption (e.g., the $\eps$-No Learning assumption for negligible $\eps$), there is a noticeable gap between the success probability of a \emph{noisy} verifier (who has the circuit description) and any \emph{noiseless} polynomial-time adversary (who doesn't have the circuit description): the adversary cannot produce any approximation to the state except with negligible fidelity. 
We can amplify this gap to obtain a NISQ-friendly OWSG, where both the generation and verification algorithms can be run on a noisy quantum computer, but the OWSG also retains its security against noiseless adversaries. We achieve this amplification via a \emph{computational Chernoff bound}.

\paragraph{Computational Chernoff bounds.} A standard way to amplify the security of a OWSG is via \emph{parallel repetition} \cite{morimae2022one, Bostanci_2024}: the key to the amplified OWSG corresponds to a $t$-tuple of independently chosen random circuits $(C_1,\ldots,C_t)$, and the output is the tensor product of the corresponding states $\ket{C_1}\otimes \cdots \otimes \ket{C_t}$. Verification proceeds by checking that the $i$'th block of $n$ qubits is in the state $\ket{C_i}$ for each $i=1,\ldots,n$. Intuitively, if it is somewhat hard for the adversary to learn the output of one random quantum circuit, then it should be \emph{very} hard for the adversary to simultaneously learn the output of \emph{many} random quantum circuits. However, because the noisy fidelity is small, standard parallel repetition theorems do not directly work because the fidelity of the overall state on $nt$ qubits would degrade exponentially with $t$, making it $\negl(n)$, which means even honest verifiers would fail to see a non-trivial signal. 
 
 Our technical innovation is to use \emph{threshold parallel repetition} for amplification---this is a \textrm{OWSG} $G^{t,k}$ consisting of $t$ independent copies of $G$, but instead of verifying that all $t$ copies have been inverted, the verification algorithm checks that at least $k$ out of the $t$ have been inverted. To bound its security we prove new \emph{computational Chernoff bounds} for OWSGs. A detailed discussion, with the theorem statements, can be found in \Cref{NISQ_OWSG}. 

\subsubsection{NISQ-friendly digital signatures} 

As a direct application of our NISQ-friendly OWSG we obtain a NISQ-friendly quantum digital signature scheme. 
At a high level, a digital signature scheme (with quantum public keys) is a method for a \emph{signer} to generate a \emph{signature} for a message in a way that a third party \emph{verifier} (using a quantum public key posted by the user beforehand) can verify that the signature belongs to the message (and in particular, the message or the signature have not been changed). 

Morimae and Yamakawa~\cite{morimae2022quantum} showed that such quantum digital signature schemes can be directly constructed from OWSGs by adapting the famous Lamport construction of digital signatures~\cite{lamport1979constructing}. We show that by plugging in the NISQ-friendly random circuit OWSG, their digital signature scheme becomes NISQ-friendly as well. This gives example of an end-to-end cryptographic task for \textrm{NISQ}-devices whose hardness relies on an innately quantum conjecture. 

We present the scheme and the analysis in detail in \Cref{sec:signature}. 



\subsubsection{On noise assumptions and asymptotics}

Note that the only noise assumption we need is that the fidelity of the signal should at most be inverse polynomially large. However, an observant reader would notice that $\mathscr{C}_n$ is an ensemble of sufficiently deep circuits. For certain noise models, for e.g. constant rate of depolarizing noise per gate, or more generally, constant rate of unital noise per gate, the output converges to the maximally mixed state \cite{Aharonov3, Deshpande_2022} exponentially fast in the depth of the circuit, which would cause inverse superpolynomially large signal decay at large depths. However, there are three perspectives on why our proposal is still relevant to near-term experiments.

Firstly, the results proving convergence to the maximally mixed state \cite{Deshpande_2022,Aharonov_2023} are asymptotic statements, whereas real experiments have finite system sizes. Thus, there is a discrepancy between what theoretically happens when we scale up the system size and what is experimentally observed for a fixed system size. For example, in quantum advantage demonstrations with random circuits, the experimentalists have observed signatures of long range entanglement and evidence of convergence to the Porter-Thomas distribution when the output state is measured in the standard basis \cite{Arute2019, morvan2023phasetransitionrandomcircuit}. Note that Porter-Thomas distribution is far in total variation distance from the uniform distribution, where the latter is what we get when we measure a maximally mixed state in the standard basis. If the system were indeed close to being maximally mixed, we would neither see long-range entanglement nor Porter-Thomas type behavior. Thus, sampling from the uniform distribution is not a good approximation to the realistic output distribution, even though the distributions are close asymptotically \cite{Aharonov3, Deshpande_2022}. 

There are other classical samplers, such as the recent one proposed in \cite{Aharonov_2023}, which differ from simply sampling the uniform distribution. However, the sampler in \cite{Aharonov_2023} achieves a smaller total variation distance than the uniform distribution only at a depth of approximately $\sim \log n$. Properties of real experiments—such as the presence of long-range entanglement—indicate that they do not operate at logarithmic depth but rather in a much deeper regime. Thus, in the same way as the uniform distribution, it is unclear if the sampler of \cite{Aharonov_2023} is directly applicable to real experiments, even though its classical spoofing distribution is again asymptotically close to the noisy target distribution.

Secondly, note that the fidelity of the output state of a noisy circuit, comprised of two qubit gates and a single-qubit, uncorrelated noise channel acting upon each qubit after the application of each gate, is 
\[
F = (1 - \epsilon)^{2s},
\]
or the probability that no errors occurred anywhere in the system. Here, where $s$ is the circuit size and $\epsilon$ is the noise rate per qubit. If $\epsilon$ is at most $\sim \frac{1}{n}$, $F \approx e^{-\Theta(s \cdot \epsilon)}$. Hence, one valid regime for inverse polynomial decay in fidelity is when $\epsilon$ is $\sim \frac{1}{n \log n}$ and the system size $s$ is $\sim n \log^2 n$. In certain models, the structure of the output state becomes even simpler. One of the a noise models for which this is manifestly true is the white noise model \cite{Arute2019, bluvstein2024logical, dalzell2024random}. According to this model, if the noise per gate is unital, and if the noise rate $\epsilon$ is at most $\sim \frac{1}{n}$, then the output state of the circuit can be written as
\[
\rho_{\mathsf{out}} = F \rho_{\mathsf{noiseless}} + (1 - F) \frac{\mathbb{I}}{2^n},
\]
that is, as a linear combination of $\rho_{\mathsf{noiseless}}$, which is what the output state would have been if there were no noise, and the maximally mixed state. While the noise rate per gate going down with $n$ is unrealistic for extremely large system sizes, it is nonetheless a reasonable model of real experiments as structural properties of experimental output states match the white noise output state. In fact, judging by the recent progress in random quantum circuit experiments, e.g.  \cite{Arute2019, morvan2023phasetransitionrandomcircuit,  Quantinuum2023}, it seems quite reasonable to model realistic noise as going down with system size.

Thirdly, note that in all of the previous discussions, we have assumed noise to be unital. But real noise is complicated and it is unclear if any of the above results (e.g., convergence to the maximally mixed state or the white noise model) are realistic for near-term experiments.  As an example of surprising behavior with more general noise models, researchers have recently studied the effects of non-unital noise channels in random quantum circuits \cite{Fefferman2024, mele2024noiseinducedshallowcircuitsabsence, oh2023classicalsimulationalgorithmsnoisy}. Non-unital noise is ubiquitous in real world experiments, as witnessed by, e.g., readout errors, $T_1$ decay for superconducting systems, and photon loss in bosonic systems \cite{QuEra2023, Xanadu2022, Zhong2021}. In particular, certain structural properties that are true for unital noise at some regimes, like anti-concentration or convergence to the maximally mixed state, fail in the presence of any constant rate non-unital noise channel \cite{Fefferman2024}, and hardness or easiness results that assume these properties, like \cite{aaronson2011computational, aharonov2003simple, Bouland_2018, Aharonov_2023}, do not seem to work.

\subsection{Related work}
\label{sec:related-work}

We discuss the relationship between our work and some concurrent, independent works that recently appeared. 

\paragraph{Cryptography from assumptions about random circuits:} Khurana and Tomer~\cite{khurana2024founding} also study the quantum cryptographic implications of hardness assumptions about random circuits. Their hardness assumption posits that it is $\#\mathrm{P}$-hard to estimate the output probabilities of a random quantum circuit, given a classical description of the circuit. This is a well-studied hardness assumption and is the theoretical basis for many quantum supremacy proposals (e.g.,~\cite{aaronson2011computational,Bouland_2018,Arute2019}). Combined with the complexity-theoretic assumption that $\mathrm{P}^{\# \mathrm{P}} \not\subseteq \mathrm{ioBQP}/\mathrm{qpoly}$,~\cite{khurana2024founding} show how to construct \emph{quantum one-way puzzles}, which in turn implies the existence of quantum bit commitments through their previous work~\cite{khurana2024commitments}.

While random circuits are at the core of the hardness assumptions of our paper as well as~\cite{khurana2024founding}, the similarities quickly end. Khurana and Tomer are positing the hardness of a \emph{classical-input, classical-output} task: given the description of a quantum circuit and a string, estimate the probability of outputting the string. Our hardness assumptions, on the other hand, are about \emph{quantum-input} tasks: given \emph{quantum} copies of the output state of a random circuit, either learn or clone it. Importantly, the circuit description is \emph{not} known to the adversary. 

Furthermore, our motivations have some differences: they are motivated by basing quantum cryptography on separations between \emph{decision} complexity classes, whereas we are primarily motivated by the connection between quantum learning problems (which involve quantum inputs) and quantum cryptography. 

On a different note, in Bostanci, Haferkamp, Hangleiter, and Poremba \cite{bostanci2024efficientquantumpseudorandomnesshamiltonian}, the authors construct quantum cryptography from a suite of assumptions about random IQP circuits. Depending on the type of assumption, the authors get quantum trapdoor functions, quantum pseudoentanglement, and candidate constructions of efficient pseudorandom unitaries. There are large differences between this work and our work, in terms of the nature of assumptions, the justifications for hardness, the flavours of cryptography that one gets, and the query lower bounds. These two works represent complementary explorations into cryptography from two different random ensembles, ours involving brickwork random circuits, and theirs involving IQP circuits.

\paragraph{Cryptography from assumptions about quantum states, protocols, and noise:} Qian, Raizes, and Zhandry~\cite{qian2024hard} study the quantum cryptographic implications of a new ``search-type'' assumption they call \emph{classical $\to$ quantum extrapolation}, where the goal is to extrapolate the rest of a bipartite pure state given the first register measured in the computational basis. They show that the hardness of this extrapolation task implies the existence of quantum bit commitments and is implied by the existence of various quantum public-key primitives. We view their work as studying the cryptographic implications of a (conceptually new) ``generic'' assumption, where they do not specify how exactly to generate the hard, inextrapolable bipartite states. On the other hand, we are focused on the cryptographic implications of a ``concrete'' assumption, where we instantiate the underlying primitive (OWSG, quantum commitment) with a concrete algorithmic implementation. This is similar to the difference between assuming that \emph{some} one-way function exists, versus assuming that a \emph{specific} one-way function exists (e.g., the RSA function or the LWE function). 

Morimae, Shirakawa, and Yamakawa~\cite{morimae2024cryptographiccharacterizationquantumadvantage} give a characterization of the complexity assumptions needed for a class of protocols for proofs of quantumness; in particular they show that one-way puzzles (the same primitive constructed by~\cite{khurana2024founding}) are necessary and sufficient. Their goal is squarely aimed at understanding the complexity of proofs of quantumness in the abstract, and less by having concrete instantiations of quantum cryptographic primitives. For related papers on the interplay between one-way puzzles, proofs of quantumness, and quantum cryptography, also see~\cite{CFP23, hiroka2024quantumcryptographymetacomplexity, cavalar2024metacomplexitycharacterizationquantumcryptography}.

Hiroka and Hsieh \cite{hiroka2024computational} studies the hardness of learning efficiently generatable pure states. The main focus of this paper is a $\mathrm{PP}$ upper bound on this task. They also base some cryptography on their hardness assumptions, like the existence of one-way state generators. The crucial difference between this work and ours is that they consider one particular learning assumption, whereas we consider a suite of learning and distinguishing assumptions, basing different flavours of cryptography on each. We also discuss in detail how to make our protocols NISQ-implementable.

Poremba, Quek, and Shor \cite{poremba2024learningstabilizersnoiseproblem} put forward a new quantum-inspired primitive called Learning Stabilizers with Noise (LSN), which deals with decoding a random stabilizer code in the presence of local depolarizing noise. Their primitive implies (statistically hiding and computationally binding) bit commitments. Their goal is to construct a new natively quantum assumption for quantum cryptography, as opposed to NISQ-friendliness of their protocols.


\paragraph{NISQ-friendly cryptography:} Finally, we comment on the relationship between the proposals by~\cite{aaronsonhung,bassirian2024certifiedrandomnessfouriersampling} to use NISQ devices to perform certifiable randomness generation. Similarly to our work, they propose a cryptographic task that can be performed on a NISQ device, and whose security can be based on hardness assumptions about random circuits. 

However a significant difference is that the verification procedure in their protocols are inherently inefficient; it requires exponential time even using a quantum computer as it requires approximating output probabilities of a random quantum circuit. In contrast, our NISQ-friendly digital signature scheme is efficiently implementable. 

\subsection{Summary}
\label{sec:disc}

Our exploration also uncovers a deeper understanding of various quantum learning tasks and cryptographic primitives.

\paragraph{Fine-grained distinctions in learning and cryptography.}
Our work connects fine-grained learning tasks to understanding the relative power of different cryptographic primitives. This opens up a potential new approach to understand both. For instance, cloning is potentially an easier task than learning\footnote{There is some suggestive evidence here in the form of a black-box separation for a related task \cite{nehoran2024computational}.}. Our work shows that the hardness of learning assumptions is essentially equivalent to the existence of $\mathrm{OWSG}$s from random circuits, while the hardness of cloning can be used to construct a fairly simple bit commitment scheme. This indicates that commitments are likely to be a more minimalistic cryptographic primitive than $\mathrm{OWSG}$s.  Our conclusion is consistent with oracular evidence in \cite{behera, bostanci2024oracle}. 


%
%

\paragraph{Practical applications of random quantum circuits.}
Random circuits have been extensively studied in the context of quantum advantage in the NISQ-era. Several experimental groups around the world
(for e.g, \cite{Arute2019, morvan2023phasetransitionrandomcircuit,Zhong2020,Zhu2021,Zhong2021} have claimed practical demonstrations of quantum advantage with sampling tasks involving random circuits. Furthermore, there also has been extensive effort to build theoretical foundations of quantum advantage based on such sampling tasks (see e.g., \cite{aaronsonchen,Bouland_2018,Aharonov_2023,Fefferman2024}), but even if we are able to demonstrate practical advantage with such circuits, one major challenge that remains is to use NISQ devices or random quantum circuits to solve practically useful problems. 

By proposing useful cryptographic applications of random circuits that are ``NISQ-friendly'', our work takes one further step in addressing this challenge and complements the recent work on certified random number generation with random circuits (\cite{aaronsonhung,bassirian2024certifiedrandomnessfouriersampling}). 


\paragraph{Minimal assumptions for cryptography.} 
Our work contributes to understanding the minimal theoretical assumptions needed for quantum cryptography. While in the classical world the existence of one-way functions is widely believed to be necessary for cryptography \cite{impagliazzo1995personal}, in the
quantum context this may not be the case. In particular, recent work has given black-box evidence in which \textrm{P = NP} and yet single-copy secure pseudorandom quantum states still exist \cite{kretschmer2023quantumcryptographyalgorithmica,kretschmer2}. This suggests that certain quantum
cryptographic primitives are possible even without the existence of one-way functions, for e.g., see \cite{ananth2022cryptographypseudorandomquantumstates, brakerski2023computational, LMW24, kretschmer2024quantumcomputableonewayfunctionsoneway}. On the other hand, in the white-box setting, all currently known constructions of such quantum pseudorandom states
require the existence of quantum-secure one-way functions~\cite{ji2018pseudorandom}.  Consequently, a major challenge that remains is to construct quantum cryptographic primitives which do not rely on the existence of one-way functions and are based on concrete hardness assumptions. Our work addresses this by constructing cryptography based on the hardness of quantum learning.



\subsection{Future directions}
\label{sec:is-this-the-future?}

The connections between the hardness of quantum learning and cryptography leads to several interesting directions that require more exploration:



\begin{itemize}
    \item \textbf{Relations between learning and cloning.} We posed two concrete assumptions about the hardness of quantum learning for random circuits. The natural open question that remains is to understand the differences between these different learning tasks (No-Learning versus No-cloning), their relative hardness, and to understand the security parameters needed for the hardness assumptions. 

    \item \textbf{Quantum cryptography from concrete hardness assumptions.} Innately quantum hardness of learning assumptions, like the No-Learning and No-Cloning assumptions, give a natural direction to give concrete instantiations of other cryptographic primitives on assumptions that might be weaker than one-way functions.


    \item \textbf{NISQ-friendly cryptography and practical applications.} While several cryptographic constructions presented in this paper are NISQ-friendly, others, such as for quantum bit commitments, involve operations that are not realistic for near-term devices, for instance, coherently implementing a superposition over all quantum circuits. This leads to the tantalizing possibility of finding other NISQ-friendly cryptographic primitives.
    



    \item \textbf{Quantum pseudorandomness and connections to complexity.} There are fundamental open questions regarding pseudorandomness properties of states produced by random quantum circuits. Using hardness of learning to probe such questions is an interesting open direction. Along this line of inquiry, one might further expect to find deeper connections between learning, cryptography, pseudorandomness, state complexity classes and circuit lower bounds.  

    \item \textbf{Realization using classical communication and quantum nodes.} There are open questions regarding whether our cryptographic protocols, like the digital signature scheme and the bit-commitment scheme, are realizable by quantum end-nodes and classical communication, obviating the need for a quantum network to transport the states. This will make the protocols even more NISQ-friendly, as noise-robust quantum networks are hard to build in practice.


\end{itemize}


\section{Preliminaries}
\label{sec:prelims}

\noindent In this section, we give an overview of our notation, collect some useful facts, theorems, and lemmas, that will be used in the rest of the paper.

\paragraph{Notation.}
We write $[t]$ to denote the set $\{1,2,\ldots,t\}$. For an integer $n$ we write $1^n$ to denote its unary representation. We write $\poly(n)$ to denote $p(n)$ for some polynomial $p$. We write $\negl(n)$ to denote a negligible function, that is, some function $\delta(n)$ such that for all polynomials $p(n)$, for all sufficiently large $n$, $\delta(n) \leq \frac{1}{p(n)}$. In other words, $\delta(n)$ goes to $0$ faster than any inverse polynomial. 

The identity operator is denoted by $\Id$. For an operator $A$ we write $\| A \|_1$ to denote its trace norm, i.e., the sum of its singular values. For two density matrices $\rho,\sigma$ we write $\mathrm{F}(\rho,\sigma) := \Tr(\sqrt{\sqrt{\rho} \sigma \sqrt{\rho}})^2$ to denote the fidelity between them.


\paragraph{Quantum circuits.}
\label{subsubsection: quantum circuits}

All quantum circuits in this paper use single- and two-qubit gates from some discrete universal gate set that includes the Clifford group, i.e., the set of unitaries generated by $\mathrm{CNOT},\mathrm{H}, \mathrm{S} = \begin{pmatrix} 1 & 0 \\ 0 & i \end{pmatrix}$. The size of a circuit is the number of gates in it. We write $\mathscr{C}_{n,d}$ to denote the set of all $n$-qubit, depth-$d$ circuits where the gates are arranged in a 1D brickwork architecture. 


We write $\mathscr{C}_n$ to denote the set $\mathscr{C}_{n,d}$ for $d = \log^2(n)$. Sometimes we will omit the subscript $n$ and write $\mathscr{C}$ when the number of qubits is clear from context. We write $C \leftarrow \mathscr{C}_n$ to denote sampling a uniformly random circuit $C$ from $\mathscr{C}_n$. For a circuit $C \in \mathscr{C}_{n,d}$, we write $\ket{C}$ to denote the state resulting from applying $C$ to the all zeroes input, i.e.,
\[
    \ket{C} := C \ket{0^n}~.
\]
For a circuit $C$, we write $\hat{C}$ to denote its \emph{classical description} (to distinguish it from the unitary operator corresponding to $C$). 

A \emph{quantum polynomial-time (QPT)} algorithm $A$ is a uniform family of circuits $\{ C_n \}_{n \in \N}$ such that there is a polynomial $p(n)$ such that the size of $C_n$ is bounded by $p(n)$ for all $n$. Here, uniform means that there is a polynomial-time classical algorithm that, on input $1^n$, outputs the classical description of $C_n$. In a QPT algorithm, we also allow the circuits to initialize some number of ancilla qubits to $\ket{0}$ and trace them out at the end of the computation (and thus a QPT algorithm in general corresponds to a quantum channel). 

\paragraph{Classical shadows.}
\label{subsection: classical shadows}
The classical shadows protocol of~\cite{Huang_2020} gives a method to perform measurements on a small number of copies of a quantum state $\rho$, and use the measurement outcomes to estimate the expectation values of $\rho$ with respect to a much larger number of observables. The method is sample efficient, but not necessarily computationally efficient. We first summarize their protocol.

\begin{longfbox}
    \begin{protocol} {\bf Classical shadows protocol} \label{prot:classical_shadows} 
    \end{protocol}

    \noindent \textbf{Parameters}: $k,t$ integers such that $t$ divides $k$. \\
    \noindent \textbf{Observables}: $A_1,\ldots,A_M$.  \\
    \noindent \textbf{Input}: $k$ copies of an $n$-qubit state $\rho$.
    
    \begin{enumerate}
        \item Sample $k$ random Clifford circuits $S_1,\ldots,S_k$. Applying $S_j$ to the $j$'th copy of $\rho$ and measure in the standard basis to obtain sample $x_j \in \{0,1\}^n$. 

        \item For each $j \in [k]$, compute the classical description of the Hermitian matrix called a \emph{shadow}:
        \[
            \hat{\rho}_j = (2^n + 1)S_j^\dagger \ketbra{x_j}{x_j} S_j - \Id. 
        \]

        \item Divide the $k$ samples into $t$ groups of $k/t$, and for each group $r \in [t]$, and each observable $i = 1,\ldots,M$, compute the following estimator:
        \[
            \hat{a}_i^{(r)} = \frac{t}{k} \sum_{j = 1 + k(r-1)/t}^{kr/t} \Tr(A_i \hat{\rho}_j)~.
        \]

        \item For each $i = 1,\ldots,M$, compute the final estimator
        \begin{equation}
            \label{eq:shadow-estimator}
            \hat{a}_i = \mathrm{median} \Big \{ \hat{a}_1^{(1)},\ldots,\hat{a}_1^{(t)}\Big \}
        \end{equation}
    \end{enumerate}
\end{longfbox}

\begin{lemma}[Performance of the classical shadows protocol]
\label{lem:classical_shadows}
    Let $\{ A_1 ,\ldots,A_M \}$ denote a set of $n$-qubit observables, i.e., each $A_j$ is a Hermitian matrix. Then the classical shadows protocol of \Cref{prot:classical_shadows} will with probability at least $1 - \delta$ produce estimates $\{ \hat{a}_1,\ldots,\hat{a}_M \}$ such that $|\hat{a}_i - \Tr(A_i \rho) | \leq \eps$ provided that 
    \[
        k \geq \frac{204}{\eps^2} \, \log \Big ( \frac{2M}{\delta} \Big) \, B \qquad \text{and} \qquad t \geq 2 \log \Big ( \frac{2M}{\delta} \Big)
    \]
    where $B = \max_i \Tr((A_i - 2^{-n} \Tr(A_i) \Id)^2)$.
    
\end{lemma}
\begin{proof}
    This is proved in the the Supplementary Information of~\cite{Huang_2020}. 
\end{proof}

\begin{corollary}
\label{corr:classical_shadows}
Let $\mathscr{C}_{n,d}$ denote the circuit ensemble described in Section \ref{subsubsection: quantum circuits}. 
There exists a quantum algorithm that, given input $\ket{C}^{\otimes k}$ where $\ket{C} = C \ket{0^n}$ for some circuit $C \in \mathscr{C}_{n,d}$, with probability at least $1 - \delta$ outputs a classical description of a circuit $D \in \mathscr{C}_{n,d}$ such that
\[
    |\langle C | D \rangle|^2 \geq 1 - \eps
\]
provided that
\[
    k \geq \mathcal{O} \left ( \frac{1}{\eps^2} \, \log \Big ( \frac{|\mathscr{C}_{n,d}|}{\delta} \Big) \right)~.
\]
Furthermore, if $\mathrm{NP}^{\mathrm{\# P}} \subseteq \mathrm{BQP}$ this task can be done in quantum polynomial time.

\end{corollary}
\begin{proof}

This follows directly from \Cref{lem:classical_shadows} where, enumerating the circuits as $\mathscr{C}_{n,d} = \{ C_1,C_2,\ldots \}$, we define the observable
\[
    A_i = C_i \ketbra{0^n}{0^n} C_i^\dagger~.
\]
The quantity $B$ in the statement of \Cref{lem:classical_shadows} can be upper-bounded by a constant $\mathcal{O}(1)$, leading to the stated sample complexity bounds. 

On input $\ket{C}^{\otimes k}$, the quantum algorithm will run the classical shadows protocol, and obtain estimates $\{ \hat{a}_i \}$. With probability $1 - \delta$ there is at least one index $i$ such that $\hat{a}_i \geq 1 - \eps$ (namely, the one corresponding to the circuit $C$ that generated the input state), so the algorithm can pick one arbitrarily (e.g. randomly selecting one) and outputting the corresponding circuit description.

We now consider the complexity of this algorithm. Observe that, as stated, the algorithm uses exponential time, simply for computing the estimates $\hat{a}_i$ for all observables $A_1,\ldots,A_M$, of which there are exponentially many. We can reformulate this algorithm so that it runs in polynomial-time, assuming that $\mathrm{NP}^{\mathrm{\# P}} \subseteq \mathrm{BQP}$. 

Consider the following decision problem: given shadows $(S_1,x_1),\ldots,(S_k,x_k)$ (i.e. descriptions of $n$-qubit Clifford circuits along with an $n$-bit string) and integers $1 \leq x < y \leq M$ represented in binary, decide if there exists an $x \leq i \leq y$ such that the corresponding estimator $\hat{a}_i$ defined in \Cref{eq:shadow-estimator} is at least $1 - \eps$. Note that each $i \in [M]$ the estimator $\hat{a}_i$ can be computed in polynomial-time given an oracle for $\mathrm{\# P}$. Therefore a nondeterministic polynomial-time Turing machine with an oracle to $\mathrm{\# P}$ can nondeterministically guess an index $i$ such that $\hat{a}_i \geq 1 - \eps$. 

Thus if $\mathrm{NP}^{\mathrm{\# P}} \subseteq \mathrm{BQP}$, the quantum algorithm can perform the shadow measurements, and then perform binary search in polynomial time to identify such an index $i$ with high probability. This concludes the ``Furthermore'' part of the corollary.
\end{proof}
\noindent Note that by the results of an independent recent work, by Hiroka and Hsieh \cite{hiroka2024computational}, the complexity theoretic inclusion can be improved to $\mathrm{PP} \subseteq \mathrm{BQP}$. 

\section{Black-box lower bounds for quantum learning and cloning}
\label{sec:blackbox}

\newcommand{\bi}{\bm{i}}

In this section, we give evidence for our hardness of learning conjectures by proving lower bounds in the black-box model that amongst other attack rule out efficient shadow tomography type attacks. We first introduce the black-box setting where we model output states of sufficiently deep random circuits by a Haar random state.

\begin{longfbox}
\textbf{Query algorithm with a state preparation oracle}
\begin{enumerate}
    \item Independently sample $N=2^{n}$ many $n$-qubit Haar random states $\CS = \{ \ket{\psi_1}, \ldots,\ket{\psi_N}\}$. Also, independently sample a uniformly random index $J \in [N]$.
    \item The quantum algorithm is given some $k = \poly(n)$ copies of the target state $\ket{\psi_J}$ and black-box access to a state preparation oracle that generates each $\ket{\psi_i}$ in the following way:
    \[ O_{\CS}\ket{i}\ket{0} =  \ket{i} \ket{\psi_i}.\]
    We can implement the above oracle unitarily by arbitrarily extending each $\ket{\psi_i}$ to an independent random basis of $(\mathbb{C}^{2})^{\otimes n}$.
    \item The output state of a $T$-query quantum algorithm in this model, just before the final measurement, can be expressed as
    \[ U_{T+1}(O_{\CS} \otimes I)\cdots (O_{\CS} \otimes I)U_2(O_{\CS} \otimes I)U_1 \ket{\psi_J}^{\otimes k}\ket{0^m},\]
    where $m= \poly(n)$ is the number of ancillas. The unitaries $U_i$'s are arbitrary fixed unitaries that do not depend on $J$ or $\CS$.
\end{enumerate}
\end{longfbox}


We now formalize our learning tasks in the black-box setting.

\begin{longfbox}
\textbf{Quantum learning task.} The algorithm succeeds if it outputs an index $i \in [N]$ such that $|\braket{\psi_i|\psi_J}|^2$ is non-negligible in $n$.

\medskip 
\textbf{Cloning task.} The algorithm succeeds if it outputs a state $\ket{\xi}$ on $n(k+1)$ qubits, such that $|\braket{\xi|\psi^{\otimes k}}|^2$ is non-negligible in $n$.
\end{longfbox}
\medskip




To motivate the above black-box setting, we note that the Haar random state models the output state of the random circuit and the random index $J$ above is the analog of the random circuit $C$.  Furthermore, the inefficient shadow tomography based algorithm to learn the circuit in the white-box setting (see \Cref{subsection: classical shadows}) has an analog in the black-box model as well. We only give a sketch of the algorithm which solves the state learning task: first take an epsilon net $\CN = \{\ket{\phi_1}, \ket{\phi_2}, \cdots, \ket{\phi_M}\}$ of the complex unit sphere in $N = 2^n$ dimensions where $\epsilon=1/\mathrm{poly}(n)$ and $M = (C/\eps)^{N}$ for some universal constant $C$. Then, using the classical shadows algorithm of \cite{Huang_2020} with the observables $\{\ketbra{\phi_i}{\phi_i}\}_{i=1}^M$, learn an index $r$ such that the overlap $|\braket{\phi_r|\psi_J}| \ge 1/\mathrm{poly}(n)$. Note that only $k=\poly(n)$ samples of the input state $\ket{\psi_J}$ are needed for this. Finally, by querying the oracle $O_{\CS}$ exponentially many times and measuring the overlap with the state $\ket{\phi_r}$, learn the index $J$. After learning the index $J$, one can solve the cloning task with non-negligible probability as well. Note that this algorithm can even be implemented \emph{non-adaptively}. 

There are several more modifications one could make to the above black-box model to make it more in line with the white-box model. For example, the choice of the parameter $N$ above could be changed to further model the fact that there are many more than $2^n$ circuits acting on $n$ qubit states. Or one could plant a small number of states that are correlated with each of the Haar random states to model the fact that the output states of circuits may have some non-trivial overlaps with a few other ones. The results proven below are robust to such changes, at least for the examples given above. We have mostly opted for the  choices made here for the purposes of presenting a cleaner analysis. 


\blackbox*




To prove the above result, it suffices to look at the cloning task, since a query-efficient algorithm for the state learning task can be used to obtain a query-efficient algorithm for the former. 

\subsection{Proof of \Cref{thm:blackbox}}

The proof proceeds in two parts: the first part follows a hybrid argument similar to \cite{BBBV} --- we show that up to a small error, the queries to the oracle can be replaced by queries to a different oracle that does not depend on $\ket{\psi_J}$. The second part then argues the base case for the algorithm which makes no queries and only uses the given copies of the input state. For the cloning lower bound, we appeal to the well-known result about the optimal cloning probability of a Haar random state \cite{werner1998optimal}. 




\paragraph{Hybrid argument.}  

Let $N=2^n$. A $T$-query algorithm starts in the initial state
\[ \ket{\phi^{(0)}} = \ket{\psi_J}^{\otimes k} \otimes \ket{0\cdots 0},\]
and the state after $t \in [T]$ queries is given by
\[ \ket{\phi^{(t)}} = (O_{\CS} \otimes I) U_{t-1} \cdots (O_{\CS} \otimes I) U_1 (O_{\CS} \otimes I) U_0\ket{\phi^{(0)}},\]
where $U_1,\ldots,U_{T+1}$ are fixed unitaries.


We will show that calls to the state preparation oracle $O_{\CS}$ can be replaced with calls to another oracle $O'_{\CS}$. Towards this end, we consider the following oracle $O'_{\CS}$:
\begin{align*}
     \ O'_{\CS}\ket{i}\ket{0} &=  \ket{i} \ket{\psi_i}, \text{ for all }i \neq J\\
     \  O'_{\CS}\ket{J}\ket{0} &=  \ket{i} \ket{\psi'},
\end{align*}
where $\ket{\psi'}$ is a Haar random state sampled independently of $J$ and $\CS$. Note that all the random variables $J, \CS = \{\ket{\psi_i}\}_{i=1}^N$ and $\ket{\psi'}$ are independent and the oracle $O'_{\CS}$ above can be implemented unitarily as before, by extending $\ket{\psi'}$ to a random basis independent of $\{\ket{\psi_i}\}_{i=1}^N$ and $J$.

Defining 
\[ \ket{{\phi'}^{(t-1)}} = (O'_{\CS} \otimes I) U_{t-1} \cdots (O'_{\CS} \otimes I) U_1 (O'_{\CS}  \otimes I) U_0\ket{\phi^{(0)}},\]
we will show that on average over the choice of $J, \CS, \ket{\phi'}$, the following holds
\begin{align}\label{eqn:ind}
    \ \BE\left[\left\|\ket{{\phi}^{(t-1)}} - \ket{{\phi'}^{(t-1)}} \right\| \right]\le \epsilon_t \text{ where } \epsilon_t = \frac{2t}{\sqrt{N}}. 
\end{align}

The statement is trivially true when no queries are made, so consider any $t \in [T]$. Then, using the triangle inequality, 
\begin{align}\label{eqn:1}
    \ \BE\left[\left\|\ket{{\phi}^{(t)}} - \ket{{\phi'}^{(t)}} \right\| \right] &\le \BE\left[\left\| \ket{{\phi}^{(t)}} - (O_{\CS} \otimes I) U_t\ket{{\phi'}^{(t-1)}} \right\| \right] \notag \\
    \ & \hphantom{shiftttt} +~ \BE\left[\left\|(O_{\CS} \otimes I) U_t\ket{{\phi'}^{(t-1)}} - \ket{{\phi'}^{(t)}} \right\| \right], 
\end{align}

Note that $\ket{{\phi}^{(t)}} = (O_{\CS} \otimes I) U_t\ket{{\phi}^{(t-1)}}$, thus by the unitary invariance of the Euclidean norm, the induction hypothesis implies that the first term is at most $\eps_{t-1}$. We now bound the second term
\begin{align*}
   \BE\left[\left\|(O_{\CS} \otimes I) U_t\ket{{\phi'}^{(t-1)}} - \ket{{\phi'}^{(t)}} \right\| \right] =  \BE\left[\left\|(O_{\CS} \otimes I) U_t\ket{{\phi'}^{(t-1)}} - (O'_{\CS} \otimes I) U_t\ket{{\phi'}^{(t-1)}} \right\| \right].
\end{align*}

Let us write
\begin{align*}
    \  \ket{\xi^{(t-1)}} := U_t\ket{{\phi'}^{(t-1)}} &=\sum_{i=1}^N \ket{i}\ket{{\xi'}_i^{(t-1)}},
\end{align*}
for some sub-normalized states $\ket{{\xi'}_i^{(t-1)}}$ satisfying $\sum_{i=1}^{N} \|\ket{{\xi'}_i^{(t-1)}}\|^2 = 1$. Note that
    \[  \left\| (O_{\CS} \otimes I) U_t\ket{{\phi'}^{(t-1)}} - (O'_{\CS} \otimes I) U_t\ket{{\phi'}^{(t-1)}}\right\| \le 2\left\|\ket{{\xi'}_J^{(t-1)}}\right\|,\]
since $O_{\CS} \otimes I $ and $O'_{\CS} \otimes I$ act the same on all states of the form $\ket{i}\ket{\phi}$ where $i \neq J$.

%
Thus, plugging this in \eqref{eqn:1}, 
\begin{align*}
    \ \BE\left[\left\|\ket{{\phi}^{(t)}} - \ket{{\phi'}^{(t)}} \right\| \right] &\le \eps_{t-1} + 2\cdot \BE\left[\left\|\ket{{\xi'_J}^{(t-1)}}\right\|\right].
\end{align*}


Note that $J$ was sampled uniformly from $[N]$ and independently of $\CS$ and $\ket{\psi'}$ which are i.i.d. Haar random states.  By symmetry, it follows that $J$ is uniformly distributed in $[N]$ conditioned on $O'_{\CS}$ and $\ket{\psi'}$. Since the state $\ket{\xi'^{(t-1)}}$ is determined by the choice $O'_{\CS}$ and $\ket{\psi'}$, it follows that
\begin{align*}
    \ \BE\left[\left\|\ket{{\phi}^{(t)}} - \ket{{\phi'}^{(t)}} \right\| \right] &\le \eps_{t-1} + \frac{2}{N}\cdot  \BE\left[\sum_{i=1}^N \left\|\ket{{\xi'}_i^{(t-1)}}\right\|\right] \notag\\
    \ &\le \eps_{t-1} + \frac{2}{N} \cdot \sqrt{N} \cdot \BE \left[\sqrt{\sum_{i=1}^N \left\|\ket{{\xi'}_i^{(t-1)}}\right\|^2} \right] \notag \\
    \ & \le \eps_{t-1} + \frac{2}{\sqrt{N}} = \frac{2t}{\sqrt{N}} = \eps_t.
\end{align*}
Note that the expectation on the right hand side is only taken over the choice of the oracle $O'_{\CS}$ and $\ket{\psi'}$.





Recalling $N=2^n$, it follows by Markov's inequality that with probability at least $1-2^{-n/4}$, the states after $T$ queries are $2T\cdot2^{-n/4}$-close in the Euclidean norm after replacing the oracle $O_{\CS}$ with $O'_{\CS}$.

\paragraph{Base case.}


Now we consider algorithms for the cloning task that receive $k$ copies of the input state $\ket{\psi_J}$ and query the oracle $O'_{\CS}$. Since $\ket{\psi_J}$ is independent of $O'_{\CS}$ and the index $J$, observe that any such quantum algorithm defines a cloning channel for a Haar random state, i.e., a quantum channel that  takes $k$ copies of an $n$-qubit Haar random state and outputs a mixed state $\rho$ on $(k+1)n$ qubits that should be close to $(k+1)$ copies of the input state. By the well-known results about the optimal cloning of Haar random states \cite{werner1998optimal}, it follows that
 \[ \mathrm{F}(\rho, \ketbra{\psi_J}{\psi_J}^{\otimes (k+1)}) \le \dfrac{\binom{2^n+k-1}{k}}{\binom{2^n+k}{k+1}} \le \frac{k+1}{2^n+k} \le k \cdot 2^{-n}. \]
Combining the hybrid argument and the base case, it follows that for the cloning task, the success probability is at most ${2T}\cdot 2^{-n/4} + 2^{-n/4} + k \cdot 2^{-n} \le 2^{-n/4} ( 2T + k + 1)$. This completes the proof of \Cref{thm:blackbox}.

\section{Cryptography from hardness of quantum learning}
In this section, we lay out the cryptographic objects we can construct from our hardness of learning conjectures.
\label{sec:crypto}

\subsection{One-way state generators}
\label{sec:OWSG}

The conjecture about the hardness of learning is essentially \emph{equivalent} to the existence of one-way state generators (OWSGs). These are a quantum analogue of one-way functions, which are functions efficiently computable in the forwards direction but computationally difficult to invert. We first recall the formal definition of a OWSG, a primitive first introduced by Morimae and Yamakawa~\cite{morimae2022quantum,morimae2022one}. 

\begin{definition}[One-way state generator]
A \emph{one-way state generator} $G$ is a pair of QPT algorithms $(\Gen,\Ver)$ such that there exists polynomials $r(n),m(n)$ such that
\begin{itemize}
    \item $\Gen$ (called the \emph{generator}) takes as input the security parameter $1^n$ in unary, a string $k \in \{0,1\}^{r(n)}$ called a \emph{key}, and outputs a $m(n)$-qubit pure quantum state $\ket{\psi_k}$.
    \item $\Ver$ (called the \emph{verification}) takes as input the security parameter $1^n$ in unary, a key $k \in \{0,1\}^{r(n)}$ and an $m(n)$-qubit state $\ket{\psi}$, and accepts or rejects. 
\end{itemize}
We say that a OWSG $G$ satisfies \emph{correctness} if for all security parameters $n$, for all keys $k \in \{0,1\}^{r(n)}$,
\[
    \Pr \Big [ \Ver(1^n,k,\Gen(1^n,k)) \textrm{ accepts} \Big ] \geq 1 - \negl(n)~.
\]
We say that \emph{$G$ has security error $\gamma(n)$} if for all polynomials $q(n)$, for all QPT algorithms $A$, for all sufficiently large $n$, 
    \[
        \Pr \left [ \Ver(1^n,k',\ket{\psi_k}) \textrm{ accepts} : \begin{array}{c} k \leftarrow \{0,1\}^{r(n)} \\ \ket{\psi_k} \leftarrow \Gen(1^n,k) \\ k' \leftarrow A(\ket{\psi_k}^{\otimes q(n)})\end{array} \right ] \leq \gamma(n)~.
    \]
    We say that $G$ is \emph{cryptographically secure} if it has security error that is negligible in $n$ (i.e., it goes to zero faster than $1/\poly(n)$). 
\end{definition}

\begin{remark}
    For simplicity we often omit the security parameter $1^n$ as an input to  $\Gen,\Ver$ when it is clear from context.
\end{remark}

Let $r(n)$ be a polynomial such that $r(n)$ bits are sufficient to describe a circuit from the family $\mathscr{C}_{n}$. We define a OWSG $G = (\Gen,\Ver)$ based on random circuits.

\begin{longfbox}
    \begin{protocol} {\bf Random circuit OWSG} \label{prot:OWSG} 
    \end{protocol}

    \textbf{$\Gen$}: Given input key $C \in \{0,1\}^{r(n)}$, interpret it as a description of an $n$-qubit circuit $C$ from the ensemble $\mathscr{C}_n$ (defined in \Cref{subsubsection: quantum circuits}). Output $\ket{C} = C \ket{0^n}$. 

    \vspace{4pt}

    \textbf{$\Ver$}: Given input $C \in \{0,1\}^{r(n)}$ and a state $\ket{D}$ on $n$ qubits, apply $C^\dagger$ to the state, and measure. Accept if the result is all zeroes, and reject otherwise. 
\end{longfbox}

It is clear that the algorithms $\Gen,\Ver$ run in polynomial time. It is also easy to see that $G$ has perfect correctness. We now argue that the security of the OWSG is essentially \emph{equivalent} to Computational No-Learning Assumption (\Cref{conj:learning}). 
The level of security corresponds to the strength of the hardness conjecture.

\begin{lemma}[Equivalence between hardness of learning and random circuit OWSG security]
\label{lem:weak-owsg}
    Assuming $\eps$-No-Learning (\Cref{conj:learning}), the random circuit OWSG has security error $(2 - \eps)\eps$. Conversely, if the random circuit OWSG has security error $\gamma$, then the $\sqrt{\gamma}$-No Learning Assumption holds.
\end{lemma}
\begin{proof}
    To prove (weak) security, assume for contradiction there exists a QPT adversary $A$ and a polynomial $q(n)$ such that for infinitely many $n$,
    \[
        \Pr \left [ \Ver(D,\ket{C}) \textrm{ accepts} : \begin{array}{c} C \leftarrow \mathscr{C}_n \\ D \leftarrow A(\ket{C}^{\otimes q(n)})\end{array} \right ] > (2 - \eps)\eps
    \]
    The acceptance probability of the verification circuit can be written as:
    \[
        \E_{D \leftarrow A(\ket{C}^{\otimes q(n)})} \left | \langle C | D \rangle \right |^{2} > (2 - \eps)\eps~.
    \]
    This implies that the probability over the choice of $D$ output by $A(\ket{C}^{\otimes q(n)})$ that $|\langle C | D \rangle|^2 \geq \eps$ is greater than $\eps$ for infinitely many $n$. This contradicts \Cref{conj:learning}. 

    Conversely, assume for contradiction that the random circuit OWSG has security error $\gamma$, but there exists a QPT algorithm $A$ such that
\[
    \Pr \left [ |\langle C | D \rangle|^2 \geq \sqrt{\gamma} : \begin{array}{c} C \leftarrow \mathscr{C}_n \\ D \leftarrow A(\ket{C}^{\otimes \poly(n)}) \end{array}  \right] > \sqrt{\gamma}~.
\]
This implies that
    \[
        \Pr \left [ \Ver(D,\ket{C}) \textrm{ accepts} : \begin{array}{c} C \leftarrow \mathscr{C}_n \\ D \leftarrow A(\ket{C}^{\otimes q(n)})\end{array} \right ] > \gamma
    \]
    which contradicts the security of the OWSG.

\end{proof}

When $\eps \geq 1/\poly(n)$, then we consider the resulting OWSG from \Cref{lem:weak-owsg} to have \emph{weak} security, as it implies that a QPT adversary could potentially produce a non-negligible approximation of $\ket{C}$ with non-negligible probability. To obtain a \emph{cryptographically secure} OWSG, we need to amplify it so that all QPT adversaries can only succeed with at most negligible probability. 

A standard way to amplify the security of a weak OWSG is via \emph{parallel repetition}: now the key to the amplified OWSG corresponds to a $t$-tuple of independently chosen random circuits $(C_1,\ldots,C_t)$, and the output is the tensor product of the corresponding states $\ket{C_1}\otimes \cdots \otimes \ket{C_t}$. Verification proceeds by checking that the $i$'th block of $n$ qubits is in the state $\ket{C_i}$ for each $i=1,\ldots,n$. Intuitively, if it is somewhat hard to learn the output of one random quantum circuit, then it should be \emph{very} hard to simultaneously learn the output of \emph{many} random quantum circuits. This intuition holds true and the parallel repetition of OWSGs is formally analyized in Morimae and Yamakawa~\cite{morimae2022one}, who showed that if a OWSG $G$ has security error $\gamma$, then the $t$-fold repetition of $G$, denoted by $G^t$, has security error $\approx \gamma^t$ (up to additive errors that are negligible in $n$). (This is also implied by the general quantum hardness amplification result of~\cite{Bostanci_2024}). 

\begin{theorem}
\label{thm:owsg}
Assuming $\eps$-No-Learning for $\eps \leq 1 - \frac{1}{\poly(n)}$, there exists a cryptographically secure one-way state generator.
\end{theorem}
\begin{proof}
    Assuming $\eps$-No-Learning, by \Cref{lem:weak-owsg} there exists a weak OWSG $G$ with security error $(2 - \eps)\eps$. Let 
    \[
    t = \frac{\log^2(n)}{\log \frac{1}{ (2 - \eps)\eps} }~.
    \]
    When $\eps \leq 1 - 1/\poly(n)$, the quantity $t$ is at most $\poly(n)$. Consider the following OWSG $\hat{G}$, which is simply the original OWSG $G$ repeated $t$ times in parallel.

\begin{longfbox}
    \begin{protocol} {\bf (Strong) one-way state generator} \label{prot:OWSG-strong} 
    \end{protocol}

    \textbf{$\Gen$}: Given input key $(C_1,\ldots,C_t) \in (\{0,1\}^{r(n)})^t$, interpret it as a description of a $t$-tuple of $n$-qubit circuits from the ensemble $\mathscr{C}_n$. Output $\ket{C_1} \otimes \cdots \otimes \ket{C_t}$.

    \vspace{4pt}

    \textbf{$\Ver$}: Given input $(C_1,\ldots,C_t) \in (\{0,1\}^{r(n)})^t$ and a state $\ket{D}$ on $nt$ qubits, apply $C_1^\dagger \otimes \cdots \otimes C_t^\dagger$ to the state, and measure. Accept if the result is all zeroes, and reject otherwise. 
\end{longfbox}

    By the result on hardness amplification of OWSGs by Morimae and Yamakawa~\cite{morimae2022one} (alternatively, by the quantum parallel repetition theorem of~\cite{Bostanci_2024}),  the security error of $\hat{G}$ is at most 
    \[
        \Big ( (2 - \eps)\eps \Big)^{t} + \negl(n) = 2^{-\log^2(n)} + \negl(n)~.
    \]
    Note that $2^{-\log^2 (n)}$ goes to $0$ faster than any inverse polynomial $1/\poly(n)$, and thus $\hat{G}$ has negligible security error.
    
\end{proof}

\subsection{Quantum bit commitments}
\label{bit commitments}
In this section we explore the cryptographic implications of our second hardness assumption, the Computational No-Cloning Assumption (\Cref{conj:no_cloning}). We show that No-Cloning directly implies the existence of secure quantum bit commitments. 

We formally define quantum bit commitment schemes. In this paper we only define a special kind known as \emph{noninteractive quantum commitments}; while quantum commitment schemes can be interactive in general, it was shown by~\cite{yan2022general} that in the quantum setting they can always be generically compiled to a simple noninteractive protocol\footnote{It is worth noting that this transformation is not generically possible in the classical setting!}. 

\begin{definition}[Noninteractive quantum bit commitment]
    A \emph{noninteractive quantum bit commitment scheme} (or a \emph{commitment scheme} for short) $\mathrm{Com}$ is a QPT algorithm that takes as input a security parameter $1^n$ and a bit $b \in \{0,1\}$, and behaves as follows: it applies a unitary $U_{n,b}$ to the all zeroes state, obtaining a bipartite pure state $\ket{\psi_b}$ on registers $\reg{AB}$. 

    We say that a commitment scheme $\mathrm{Com}$ satisfies \emph{correctness} if for all security parameters $n$, for all $b \in \{0,1\}$, 
    \[
        |\langle \psi_0 | \psi_1 \rangle|^2 \leq \negl(n)
    \]
    where $\ket{C_b}$ is the output of $\mathrm{Com}(1^n,b)$. 
    We say that $\mathrm{Com}$ satisfies \emph{$\eps$-statistical hiding} if 
    \[
        \mathrm{F}(\rho_0,\rho_1) \geq 1 - \eps(n)
    \]
    where $\mathrm{F}(\cdot,\cdot)$ is the fidelity between two density matrices, and $\rho_b$ is the reduced density matrix of $\ket{\psi_b}$ on register $\reg{B}$. 
    We say that $\mathrm{Com}$ satisfies \emph{$\delta$-computational binding} if for all QPT adversaries $A$, for all sufficiently large $n$, 
    \[
        \mathrm{F}( \ketbra{\psi_0}{\psi_0}, (A_n \otimes \Id)(\ketbra{\psi_1}{\psi_1}) ) \leq \delta(n)
    \]
    where $A_n$ denotes running $A$ with security parameter $1^n$ and it takes as input register $\reg{A}$ of $\ket{\psi_1}$. 
\end{definition}
For the remainder of this section, we will simply refer to noninteractive commitment schemes as simply a commitment scheme. We only defined the notion of commitments with statistical hiding and computational binding; there is also the other ``flavor'' of commitments such as computational hiding and statistical binding. These flavors can be efficiently switched in a blackbox way; see~\cite{yan2022general,hhan2023hardness} for a proof. 



We note that at least one of the hiding or binding properties must rely on computational hardness assumptions~\cite{brassard1997brief}; in other words, there do not exist quantum commitment schemes that are both statistically hiding as well as statistically binding. For an exploration of the complexity-theoretic underpinnings of the security of quantum commitment schemes, see Bostanci, et al.~\cite{bostanci2023unitary}.

\vspace{8pt}

We now define a commitment scheme based on random circuits. Recall that $\mathscr{C}_n$ is the ensemble of $n$-qubit quantum circuits as defined in \Cref{subsubsection: quantum circuits}. To distinguish between a classical description of a circuit and the state generated by the circuit, we use the following notation: $\ket{\hat{C}}$ denotes the standard basis state of $r(n) = \log |\mathscr{C}_n|$ qubits that represents the classical description of the circuit $C$, and $\ket{C} = C\ket{0^n}$ denotes the state output by the circuit on the all zeroes input. In what follows, we set $k = \cO(\eps^{-2} \log |\mathscr{C}_n|/\eps) = \cO(n \log^2 n)$ for some $\eps = 1/\poly(n)$ to be determined later.

\begin{longfbox}
    \begin{protocol} {\bf Commitment scheme based on random circuits} \label{prot:commitment} 
    \end{protocol}

    $\mathrm{Com}(1^n,b)$: If $b = 0$, then prepare the state 
    \[
        \ket{\psi_0}_{\reg{AB}} := \frac{1}{\sqrt{|\mathscr{C}_n|}} \sum_{C \in \mathscr{C}_n} \Big( \ket{C}^{\otimes k} \otimes \ket{0^n} \Big)_{\reg{A}} \otimes \ket{\hat{C} }_{\reg{B}}~.
    \] 
    If $b = 1$, then prepare the state
    \[
        \ket{\psi_1}_{\reg{AB}} := \frac{1}{\sqrt{|\mathscr{C}_n|}} \sum_{C \in \mathscr{C}_n}  \ket{C}^{\otimes (k+1)}_{\reg{A}} \otimes \ket{\hat{C} }_{\reg{B}}~.
    \] 
    
\end{longfbox}

It should be clear that for each $b \in \{0,1\}$ there is an efficiently computable unitary $U_{n,b}$ that prepares $\ket{\psi_b}$ from the all zeroes state. Therefore $\mathrm{Com}$ is a QPT algorithm. 

We argue the correctness property of $\mathrm{Com}$.

\begin{claim}
    \label{clm:commitment-correctness}
    The commitment scheme of \Cref{prot:commitment} satisfies correctness.
\end{claim}
\begin{proof}
    We evaluate the overlap between $\ket{\psi_0},\ket{\psi_1}$:
    \begin{align*}
        |\langle \psi_0 | \psi_1 \rangle|^2 &= \Big | \frac{1}{|\mathscr{C}_n|} \sum_{C \in \mathscr{C}_n} \langle C | 0^n \rangle \Big|^2 \\
        &\leq \Big( \frac{1}{|\mathscr{C}_n|} \sum_{C \in \mathscr{C}_n} |\langle C | 0\rangle| \Big)^2 \\
        &\leq \frac{1}{|\mathscr{C}_n|} \sum_{C \in \mathscr{C}_n} |\langle C | 0\rangle|^2 \\
        &= \Tr \Big( \frac{1}{|\mathscr{C}_n|} \sum_{C \in \mathscr{C}_n} \ketbra{C}{C} \cdot \ketbra{0^n}{0^n} \Big) \\
        &= \frac{1}{2^n}
    \end{align*}
    where we used the property that the uniform distribution over $\mathscr{C}_n$ forms a $1$-design; here we use the fact that the gate set includes all Cliffords and therefore all Pauli operators. The uniform distribution over $\mathscr{C}_n$ is then invariant under applying a layer of random Pauli operators on each qubit at the end of each circuit. 
\end{proof}

We now analyze the statistical hiding property of the commitment scheme.

\begin{claim}
\label{clm:hiding}
    The commitment scheme of \Cref{prot:commitment} satisfies $4\eps$-statistical hiding.
\end{claim}
\begin{proof}
    Let $\rho_0,\rho_1$ denote the reduced density matrices of $\ket{\psi_0},\ket{\psi_1}$ on register $\reg{B}$. To bound the fidelity between $\rho_0,\rho_1$, we use Uhlmann's theorem~\cite{Uhlmann1976}, which implies that for all unitary operators $U$ acting on register $\reg{A}$ and an ancilla register $\reg{E}$ (consisting of $r$ qubits), we have
    \[
        \mathrm{F}(\rho_0,\rho_1) \geq | \bra{0^r,\psi_0}_{\reg{EAB}} \, (U_{\reg{EA}} \otimes \Id_{\reg{B}}) \, \ket{0^r,\psi_1}_{\reg{EAB}} |^2~.
    \]
    Here, $\ket{0^r,\psi_b\rangle}$ is simply shorthand for pre-pending $r$ ancilla qubits to the state $\ket{\psi_b}$.

To give a lower bound it suffices to describe \emph{some} local unitary operator $U$ that, with some ancillas, maps $\ket{\psi_0}$ to have large overlap with $\ket{\psi_1}$. At a high level, the unitary $U$ will coherently run the classical shadows protocol of Huang, Kueng, and Preskill~\cite{Huang_2020} as described in \Cref{subsection: classical shadows} on the $k$ copies of $\ket{C}$ in order to obtain a classical description of a circuit $D$ such that $\ket{D} \approx \ket{C}$. Controlled on this description, a copy of $\ket{D}$ is synthesized. Since $\ket{D}$ is close to $\ket{C}$ with very high probability, the intermediate work of the classical shadows protocol can be uncomputed with high fidelity. 

In more detail, the unitary $U$ behaves as follows. Given input $\ket{C}^{\otimes k}$ (plus some ancillas), the circuit-learning algorithm from \Cref{corr:classical_shadows} can be purified to be a unitary operator $V$ that is run coherently to yield
\[
    V (\ket{0^r} \otimes \ket{C}^{\otimes (k-1)}) \otimes \ket{C} \otimes \ket{0^n} = \sum_{D \in \mathscr{C}_n} \sqrt{p_{C,D}} \ket{\vartheta_{C,D}} \otimes \ket{\hat{D}} \otimes \ket{C} \otimes \ket{0^n}
\]
where $p_{C,D}$ is the probability that running the algorithm $V$ on $\ket{C}^{\otimes (k-1)}$ outputs $D$, the state $\ket{\hat{D}}$ is the classical description of $D$, and the state $\ket{\vartheta_{C,D}}$ is the post-measurement state after measuring $\hat{D}$. Controlled on $\ket{\hat{D}}$, the inverse $D^\dagger$ can be applied to the remaining copy of $\ket{C}$, and a copy of $\ket{D}$ can be synthesized to yield
\[
    \sum_{D \in \mathscr{C}_n} \sqrt{p_{C,D}} \ket{\vartheta_{C,D}} \otimes \ket{\hat{D}} \otimes D^\dagger \ket{C} \otimes \ket{D}~.
\]
The last two $n$-qubit registers are swapped to yield:
\[
    \sum_{D \in \mathscr{C}_n} \sqrt{p_{C,D}} \ket{\vartheta_{C,D}} \otimes \ket{\hat{D}} \otimes \ket{D} \otimes D^\dagger \ket{C}~.
\]
Finally, applying the inverse $V^\dagger$ we get
\[
    \sum_{D \in \mathscr{C}_n} \sqrt{p_{C,D}} V^\dagger(\ket{\vartheta_{C,D}} \otimes \ket{\hat{D}}) \otimes \ket{D} \otimes D^\dagger \ket{C}~.
\]
This concludes the description of the unitary $U$. 
We can now compute the inner product between $U \ket{\psi_0,0 \cdots 0}$ and $\ket{\psi_1,0 \cdots 0}$. We rearrange the ancillas for notational convenience:
\begin{align*}
    &\bra{0^r,\psi_1,0^n} \Big( \frac{1}{\sqrt{|\mathscr{C}_n|}} \sum_{C,D \in \mathscr{C}_n}  \sqrt{p_{C,D}} V^\dagger(\ket{\vartheta_{C,D}} \otimes \ket{\hat{D}}) \otimes \ket{D} \otimes D^\dagger \ket{C} \Big) \otimes \ket{\hat{C}} \\
    \qquad &=  \frac{1}{|\mathscr{C}_n|} \sum_{C, D \in \mathscr{C}_n} \sqrt{p_{C,D}} (\bra{0^r} \otimes \bra{C}^{\otimes k}) V^\dagger (\ket{\vartheta_{C,D}} \otimes \ket{\hat{D}}) \, \langle C | D \rangle \, \bra{0^n} D^\dagger \ket{C} \\
    \qquad &= \frac{1}{|\mathscr{C}_n|} \sum_{C, D \in \mathscr{C}_n} p_{C,D}  \, |\langle C | D \rangle|^2~.
\end{align*}
Note that $\sum_{D \in \mathscr{C}_n} p_{C,D} |\langle C | D \rangle|^2$ is the expected overlap between the state of the circuit $D$ output by the classical shadows protocol and the state $\ket{C}$; by \Cref{corr:classical_shadows} this is at least $(1 - \eps)^2$. Thus the fidelity between $\rho_0,\rho_1$ is at least $(1 - \eps)^4 \geq 1 - 4\eps$. This concludes the proof of \Cref{clm:hiding}. 

\end{proof}

Finally, we prove that the commitment scheme is secure under the No-Cloning Assumption.

\begin{claim}
    \label{clm:binding}
    The $\delta$-No-Cloning Assumption (\Cref{conj:no_cloning}) implies that the commitment scheme from \Cref{prot:commitment} satisfies $(2 - \delta)\delta$-computational binding.
\end{claim}
\begin{proof}
    Suppose there was a QPT adversary $A$ that for infinitely many $n \in \N$,
    \[
        \mathrm{F}( \ketbra{\psi_0}{\psi_0}, (A_n \otimes \Id)(\ketbra{\psi_1}{\psi_1}) ) > 2\delta(n)~.
    \]
    Purifying the algorithm $A$ to include ancillas, there exists a unitary $V$ and a pure state $\ket{\vartheta}$ such that
    \[
        |\bra{\vartheta,\psi_0} (V \otimes \Id) \ket{0^r,\psi_1} |^2 \geq (2 - \delta(n))\delta(n)~.
    \]
    This however means that
    \[
        \frac{1}{|\mathscr{C}_n|} \sum_{C \in \mathscr{C}_n} \Big | (\bra{\vartheta} \otimes \bra{C}^{\otimes (k+1)}) V (\ket{0^r} \otimes \ket{C}^{\otimes k} \otimes \ket{0^n}) \Big |^2 > (2 - \delta(n)) \delta(n)
    \]
    by Jensen's inequality. This implies that the algorithm $A$ satisfies, for infinitely many $n$,
\[
    \Pr \left [ |\langle C |^{\otimes (k+1)} |\phi \rangle|^2 \geq \delta(n) : \begin{array}{c} C \leftarrow \mathscr{C}_n \\ \ket{\phi} \leftarrow A(\ket{C}^{\otimes k}) \end{array}  \right] > \delta(n)
\]    
which contradicts the No-Cloning Assumption. 
\end{proof}

\paragraph{A converse?} Does the computational binding security of $\mathrm{Com}$ \emph{imply} the No-Cloning Assumption? Intuitively, an efficient algorithm to clone outputs of random circuits (i.e., break the No-Cloning Assumption) should be able to break the binding property of the commitment scheme $\mathrm{Com}$. However, this requires that the cloning algorithm can be run \emph{coherently}, without ancillary junk states that are entangled with the the underlying circuit $C$ -- this is because breaking binding requires coherently mapping $\ket{C}^{\otimes k}$ to as close to $\ket{C}^{\otimes (k+1)}$ as possible. This is not \emph{a priori} guaranteed by the fact that the No-Cloning Assumption was broken. 

\vspace{8pt}

We summarize the conclusions of Claims~\ref{clm:commitment-correctness},~\ref{clm:hiding}, and~\ref{clm:binding} below. 

\begin{lemma}
\label{lem:commitments}
The commitment scheme of \Cref{prot:commitment} satisfies correctness, $4\eps$-statistical hiding, and (assuming $\delta$-No-Cloning) $(2 -\delta)\delta$-computational binding.
\end{lemma}

One could ask whether the security guarantees of this commitment scheme could be improved. Just like how we were able to amplify a weakly-secure OWSG to a cryptographically secure OWSG via parallel repetition in \Cref{sec:OWSG}, we would like to amplify the commitment scheme of \Cref{prot:commitment} so that it has negligible error for both the hiding and binding properties. 

Amplification of bit commitments is more subtle than with OWSGs, however, because there are two different security properties to handle. In general, trying to amplify one security property comes at the cost of degrading the other security property. For example, it was only recently shown by Bostanci, Qian, Spooner, and Yuen~\cite{Bostanci_2024} that the $t$-fold parallel repetition of a $\delta$-computational binding quantum commitment $\mathrm{Com}$ yields a new commitment scheme $\mathrm{Com}^t$ with roughly $\delta^t$-computational binding. However, if the original commitment $\mathrm{Com}$ satisfied $\eps$-statistical hiding, then $\mathrm{Com}^t$ satisfies $t\eps$ hiding -- it is now \emph{easier} for the receiver to distinguish between commitments to $0$ and $1$ (because now it has $t$ chances to do so). This would be fine if $\eps$ were a negligible quantity to begin with (i.e., the base commitment scheme had negligible hiding error), but in our case $\eps$ is at best an inverse polynomial quantity (this is due to the sample complexity of the classical shadows protocol). 



To perform the amplification, we take advantage of the ``flavor switching'' transformation for quantum bit commitments, which allows us to take a quantum commitment of one flavor (e.g., statistical hiding, computational binding), and generically obtain a quantum commitment of the other flavor (e.g., computational hiding, statistical binding) with only a small loss in parameters. 

\begin{lemma}[Flavor switching]
    \label{lem:flavor-switching} If $\mathrm{Com}$ is a $\eps$-statistical (resp. computational) hiding and $\delta$-computational (resp. statistical) binding quantum commitment, then there exists a commitment $\mathrm{Com}'$ that satisfies $\sqrt{\delta}$-computational (resp. statistical) hiding and $\eps$-statistical (resp. computational) binding.
\end{lemma}
\begin{proof}
    This is proved in~\cite[Theorem 7]{hhan2023hardness}. 
\end{proof}

With this in hand we obtain the following amplification result:
\begin{lemma}
    Assuming $\delta$-No-Cloning for some $\delta(n) \leq 1 - 1/p(n)$ for some polynomial $p$, there exists a cryptographically-secure quantum commitment scheme satisfying correctness, $\negl(n)$-statistical hiding and $\negl(n)$-computational binding. 
\end{lemma}
\begin{proof}
First, we choose $\eps(n) = (2np(n))^{-2}$ where $p$ satisfies $\delta(n) \leq 1 - \frac{1}{p(n)}$. Then assuming $\delta$-No-Cloning, \Cref{lem:commitments} implies that the commitment scheme $\mathrm{Com}$ of \Cref{prot:commitment} (with the $k(n)$ parameter chosen as a function of $\eps(n)$)  has $4\eps$-statistical hiding and $(2 - \delta)\delta$-computational binding. 

Let $\mathrm{Com}'$ denote the $n p(n)^2$-fold parallel repetition of the commitment $\mathrm{Com}$, which means that the committer and receiver run $ n p(n)^2$ parallel independent instances of $\mathrm{Com}$. The parallel repetition theorem of Bostanci, Qian, Spooner, and Yuen~\cite{Bostanci_2024} implies that $\mathrm{Com}'$ has $\frac{1}{n}$-statistical hiding and the computational binding security error is at most
\[
 ((2 - \delta(n))\delta(n))^{n p(n)^2} + \negl(n) \leq \Big (1 - \frac{1}{p(n)^2} \Big )^{n p(n)^2} + \negl(n) \leq e^{-\Omega(n)} + \negl(n) = \negl(n)~.
\]
Switching flavors (using \Cref{lem:flavor-switching}), we get a commitment $\mathrm{Com}''$ with $\negl(n)$-computational hiding and $\frac{1}{n}$-statistical binding. We then perform parallel repetition once again, repeating the commitment $\mathrm{Com}''$ for $n$ times in parallel to obtain a commitment $\mathrm{Com}'''$ where the computational hiding security error is at most $n \cdot \negl(n) = \negl(n)$ and the statistical binding security error is at most $n^{-n} = \negl(n)$. The computational hiding bound is argued via a hybrid argument: if there were an efficient algorithm $A$ that could distinguish between $\rho_0^{\otimes n}$ and $\rho_1^{\otimes n}$ with non-negligible advantage $\alpha$ where $\rho_0,\rho_1$ are the reduced density matrices seen by the receiver in $\mathrm{Com}''$, then via the triangle inequality there exists a $j \in [n]$ such that $A$ can distinguish between $\rho_0^{\otimes j} \otimes \rho_0 \otimes \rho_1^{\otimes (n-j-1)}$ and $\rho_0^{\otimes j} \otimes \rho_1 \otimes \rho_1^{\otimes (n-j-1)}$ with advantage at least $\alpha/n$, which contradicts the negligible hiding security error of $\mathrm{Com}''$. 

The statistical binding bound is argued by viewing the binding security game of a commitment scheme as a $2$-message interactive protocol, and using the fact that parallel repetition reduces the soundness error of such protocols at an exponential rate~\cite[Theorem 6]{kitaev2000parallelization}. 

Finally, we can switch flavors once more to obtain the final commitment $\mathrm{Com}''''$, which satisfy $\negl(n)$-statistical hiding and $\negl(n)$-statistical binding.
\end{proof}

\section{NISQ-friendly quantum cryptography}

\subsection{NISQ-friendly one-way state generators}
\label{NISQ_OWSG}
We take a step further and ask whether it is possible to obtain a OWSG that is both cryptographically secure and \emph{NISQ-friendly}. NISQ-friendliness of a quantum cryptographic primitive generally means that the correctness property of the primitive still holds even when the quantum algorithms in the primitive (e.g., the key generation or verification algorithms) suffer from noise. We define this notion for a OWSG.

\begin{definition}[Noise-robust OWSG]
    Let $\cN$ denote a noise model for quantum computers. We say that a OWSG $G =(\mathrm{Gen},\mathrm{Ver})$ is \emph{$\eta(n)$-robust against noise model $\cN$} if for all security parameters $n$, for all keys $k \in \{0,1\}^{r(n)}$,
\[
    \Pr \Big [ \widetilde{\Ver}(1^n,k,\widetilde{\Gen}(1^n,k)) \textrm{ accepts} \Big ] \geq \eta(n)
\]
    where $\widetilde{\mathrm{Gen}}$ and $\widetilde{\mathrm{Ver}}$ denote the quantum channels corresponding to running the algorithms $\mathrm{Gen},\mathrm{Ver}$ on a quantum computer with noise model $\cN$. 
\end{definition}

In realistic devices, noise rapidly degrades the fidelity of the signal. The fidelity between the states produced by the noisy circuit and an ideal circuit could be $1/p(n)$ for some polynomial $p(n)$, which implies $\eta(n)$ is also inverse polynomially small. This gives adversaries a greater leeway to break the security of the construction as generating a state that only has inverse polynomial overlap with the output is enough to pass the verification step. In this work, we show how to amplify the security of the one way state generator even in this case. As discussed in \Cref{sec: intro}, there are experimentally realistic depth regimes and noise rates for which inverse polynomial fidelity is reasonable due to the white noise phenomenon---for instance, see \cite{Arute2019, morvan2023phasetransitionrandomcircuit, dalzell2024random}.

What can we say about the security of the the OWSG in this regime? We have not changed the OWSG construction, so the security guarantee still holds with respect to \emph{any} polynomial-time adversary, including noise-free ones. Suppose we assume $\eps$-No-Learning for $\eps \ll 1/p(n)$. This assumption is consistent with what we know about the complexity of the learning task; as discussed in \Cref{section: intro_learning assumptions}, we do not know of an efficient algorithm that can, given polynomially-many copies of the state $\ket{C}$, to output a circuit description $D$ that has fidelity any better than $2^{-\Omega(n)}$. 

Even though the noisy quantum computer can only verify the outputs of the OWSG with small probability (even when given the key), any efficient adversary -- even a noise-free one -- has a \emph{much smaller} probability of being able to invert the outputs of a OWSG. We now have an exploitable gap: we can amplify a less noise robust OWSG to have high noise robustness, yet preserve security.

Parallel repetition is no longer a good amplification technique: although the amplified OWSG $G^t$ is secure against polynomial time adversaries, it may not be possible for honest adversaries to successfully run the verification procedure of $G^t$ on a noisy quantum computer: if the success probability of a single verification of $G$ is at $\eta$, then the success probability of $t$ parallel verifications is $\eta^t$, an exponentially small quantity. 

We instead take the \emph{threshold repetition} of $G = (\Gen,\Ver)$; this is a OWSG $G^{t,k}$ consisting of $t$ independent copies of $G$, but instead of verifying that all $t$ copies have been inverted, the verification algorithm checks that at least $k$ out of the $t$ have been inverted. We formally define the threshold repetition OWSG $G^{t,k} = (\Gen^t,\Ver^{t,k})$ next. Let $\ket{\psi_k}$ denote the output of $\Gen$ on input $k$ (we omit mention of the security parameter $n$ for convenience). 

\begin{longfbox}
    \begin{protocol} {\bf Threshold one-way state generator $G^{t,k}$} \label{prot:OWSG-threshold} 
    \end{protocol}

    \textbf{$\Gen^t$}: Given input $(k_1,\ldots,k_t) \in (\{0,1\}^{r(n)})^t$, output $\ket{\psi_{k_1}} \otimes \cdots \otimes\ket{\psi_{k_t}}$.

    \vspace{4pt}

    \textbf{$\Ver^{t,k}$}: Given input $(k_1,\ldots,k_t) \in (\{0,1\}^{r(n)})^t$ and a state $\ket{D}$ on $nt$ qubits, for all $i \in [t]$, run the verification procedure $\Ver$ on the $i$'th block of $n$ qubits with input $k_i$. Accept iff at least $k$ of the individual verifications accept. 
\end{longfbox}

We now show that for an appropriate choice of $t,k$, the threshold repetition $G^{t,k}$ has good noise robustness and also good security. This proof relies on two types of Chernoff bounds. One is the standard one (that the sum of independent random variables concentrates around their mean); this is used to obtain the improved noise robustness. The other is a \emph{computational Chernoff bound}, which argues that if an efficient adversary has at most $\gamma$ probability of inverting the output of a OWSG, then an efficient adversary has an exponentially small probability of inverting significantly more than $\gamma$ fraction of $t$ independent instances of the OWSG. In theoretical computer science and cryptography, such a result is also known as a \emph{threshold direct product theorem} (see, e.g., ~\cite{impagliazzo2010constructive}). 

\begin{lemma}
\label{lem:threshold}
    Let $\cN$ denote a noise model for quantum computers such that independent, parallel computations (i.e., they do not share qubits) experience independent noise. Let $G = (\Gen,\Ver)$ be a OWSG that is $\eta(n)$-robust against $\cN$ and has security error $\gamma(n)$ such that $\eta(n) - \gamma(n) \geq 1/p(n)$ for some polynomial $p(n)$. Then for all sufficiently large polynomials $t(n)$, the threshold repetition $G^{t,k}$ for $k(n) = \Big( \eta(n) - \sqrt{\frac{n}{2t(n)}} \Big) t(n)$ is $(1 - \cO(2^{-n}))$-robust against $\cN$ and has negligible security error.
\end{lemma}
\begin{proof}
We sometimes omit mention of $n$ for notational clarity, and write $t = t(n), k = k(n), \eta = \eta(n)$, etc.

We upper bound the probability that fewer than $k$ out of $t$ of the verifications fail to accept in $G^{t,k}$. By the assumption on the noise model in the lemma statement, the acceptance of each verification is a Bernoulli random variable $X_i$ with bias at least $\eta$. Therefore by the Chernoff-Hoeffding bound, 
\[
    \Pr \left [ X_1 + \cdots + X_t < k \right] = \Pr \left [ X_1 + \cdots + X_t < \eta t - \sqrt{\frac{nt}{2}} \right] \leq 2\exp \left ( -\Omega(n) \right)~.
\]
This establishes the noise robustness of $G^{t,k}$. We now argue about its security, using the following computational Chernoff bound:

\begin{restatable}[Computational Chernoff bound for OWSGs]{lemma}{chernoff}
\label{lem:chernoff}
Let $\xi(n)$ denote an inverse polynomial, i.e., $\xi(n) = 1/p(n)$ for some polynomial $p(n)$.
    Let $G = (\Gen,\Ver)$ be a OWSG with security error $\gamma(n)$. Then for all sufficiently large polynomials $t(n)$, the threshold OWSG $G^{t,k}$ has negligible security error, where $k(n) = (\gamma(n) + \xi(n)) t(n)$.
\end{restatable}

We prove \Cref{lem:chernoff} in \Cref{sec:chernoff}. 
Let $\xi(n) := \frac{k(n)}{t(n)} - \gamma(n)$. Then
\begin{align*}
    \xi(n) &= \eta(n) - \sqrt{\frac{n}{2t(n)}} - \gamma(n) \\
    &\geq \frac{1}{p(n)} - \sqrt{\frac{n}{2t(n)}}
\end{align*}
by our assumption on the gap $\eta(n) - \gamma(n)$. For sufficiently large polynomials $t(n)$, this is at least $1/p'(n)$ for some other polynomial $p'(n)$. The conditions of \Cref{lem:chernoff} are satisfied, and therefore for sufficiently large polynomial $t(n)$, the threshold repetition $G^{t,k}$ has negligible security error. 

\end{proof}

Combining this with our Computational No-Learning assumption, we obtain the following:

\begin{corollary}[NISQ-friendly random circuit OWSG]
\label{cor:robust-owsg}
Let $\cN$ denote a noise model where for some polynomial $p$, a $n$-qubit, depth-$d$ circuit $C$ can be run on the noisy quantum computer with fidelity at least $1/p(nd)$. Assuming $\negl(n)$-No-Learning, for a sufficiently large polynomial $t(n)$, setting $k(n) = \Big(\frac{1}{p(nd)} - \sqrt{\frac{n}{2t(n)}} \Big) t(n)$, the threshold repetition $G^{t,k}$ of the random circuit OWSG $G$ from \Cref{prot:OWSG} is $(1-\cO(2^{-n}))$-robust against $\cN$ and has negligible security error. 
\end{corollary}

\subsection{NISQ-friendly quantum digital signatures}
\label{sec:signature}
Although one-way functions are a fundamental primitive in classical cryptography, they are not very useful by themselves: their utility comes from being building blocks within cryptographic protocols such as encryption, or pseudorandomness generation~\cite{goldreich2001foundations}. Similarly, one-way state generators are useful as building blocks within \emph{quantum} cryptographic protocols, such as bit commitments~\cite{khurana2024commitments} or digital signatures~\cite{morimae2022quantum}. Thus, we would like to realize the utility of a NISQ-friendly OWSG by using it to obtain a NISQ-friendly quantum cryptographic protocol that is amenable to a real world implementation.

We illustrate this possibility with a NISQ-friendly \emph{quantum digital signature scheme}. At a high level, a digital signature scheme is a method for a user to generate a \emph{signature} for a message in a way that a third party (using a public key posted by the user beforehand) can verify that the signature belongs to the message (and in particular, the message or the signature have not been changed). While it has been long known that digital signatures are constructible from one-way functions~\cite{lamport1979constructing}, Morimae and Yamakawa~\cite{morimae2022quantum} showed that one-way functions are not necessary, and one can use a ``fully quantum'' primitive instead -- namely, one-way state generators. 



Here we instantiate the Morimae-Yamakawa digital signature construction with the NISQ-friendly random circuit OWSG from \Cref{cor:robust-owsg}. The security of the digital signature scheme follows directly from the security analysis of~\cite{morimae2022quantum} and the Computational No-Learning assumption. We furthermore argue that since the underlying OWSG is noise robust, so is the digital signature scheme, meaning that it can be implemented on noisy quantum computers. 

We first present the formal definition of a signature scheme with quantum public keys:

\begin{definition}
    A \emph{signature scheme with quantum keys} is a tuple of algorithms \\ $(\mathrm{SKGen},\mathrm{PKGen},\mathrm{Sign},\mathrm{Ver})$ satisfying the following:
    \begin{enumerate}
        \item The classical randomized polynomial-time algorithm $\mathrm{SKGen}$ takes as input a security parameter $1^n$, and then outputs a secret key $\mathrm{sk}$, which is a classical string.
        \item The QPT algorithm $\mathrm{PKGen}$ takes as input a secret key string $\mathrm{sk}$, and deterministically outputs a quantum public-key state $\ket{\mathrm{pk}}$. 

        \item The classical randomized polynomial-time algorithm $\mathrm{Sign}$ takes as input a secret key $\mathrm{sk}$ and a message $m$ and outputs a classical signature $\sigma$. 

        \item The QPT algorithm $\mathrm{Ver}$ takes as input a quantum public key $\ket{\mathrm{pk}}$, a message $m$ and a candidate signature $\sigma$ and accepts or rejects. 
    \end{enumerate}    
    We say that such a signature scheme satisfies \emph{correctness} if for all security parameters $n$, for all messages $m$, 
    \[
        \Pr \left [ \mathrm{Ver}(\ket{\mathrm{pk}}, m, \sigma) \text{ accepts} :  \begin{array}{c} \mathrm{sk} \leftarrow \mathrm{SKGen}(1^n) \\ \ket{\mathrm{pk}} \leftarrow \mathrm{PKGen}(\mathrm{sk}) \\ \sigma \leftarrow \mathrm{Sign}(\mathrm{sk}, m) \end{array} \right] = 1 - \negl(n)~.
    \]
\end{definition}

We now define a notion of \emph{one-time security} for a signature scheme with quantum public-keys. Intuitively, the security definition stipulates that a polynomial-time adversary, given copies of the quantum public-key $\ket{\mathrm{pk}}$, a message $m$ (which the adversary can choose) and its corresponding signature $\sigma$, cannot produce a valid message-signature pair $(m',\sigma')$ for some message $m' \neq m$ with non-negligible probability. 

\begin{definition}
    A signature scheme with quantum public keys $(\mathrm{SKGen},\mathrm{PKGen},\mathrm{Sign},\mathrm{Ver})$ satisfies \emph{one-time security} if for all polynomials $p(n)$, for all pairs of QPT algorithms $A_1,A_2$, the following holds for sufficiently large $n$:
    \[
        \Pr \left [ m' \neq m \wedge \mathrm{Ver}(\ket{\mathrm{pk}}, m', \sigma') \text{ accepts} : \begin{array}{c} \mathrm{sk} \leftarrow \mathrm{SKGen}(1^n) \\ \ket{\mathrm{pk}} \leftarrow \mathrm{PKGen}(\mathrm{sk}) \\ m \leftarrow A_1 (\ket{\mathrm{pk}}^{\otimes p(n)}) \\ \sigma \leftarrow \mathrm{Sign}(\mathrm{sk}, m) \\ 
        (m',\sigma') \leftarrow A_2(\ket{\mathrm{pk}}^{\otimes p(n)}, m, \sigma) \end{array} \right ] = \negl(n)
    \]
\end{definition}

Operationally, a signature scheme with quantum keys is used as follows: there is one party called the \emph{signer} and many \emph{verifiers}. First, the signer will generate a secret key $\mathrm{sk}$ and many copies of the public key $\ket{\mathrm{pk}}$. The signer publishes the quantum public keys on some central website on the (quantum) internet. 

To sign a message $m$, the signer computes the classical signature $\sigma \leftarrow \mathrm{Sign}(\mathrm{sk},m)$ and publishes the message/signature pair $(m,\sigma)$ on the (classical) internet. To verify that the signature is valid, anyone with a copy of $\ket{\mathrm{pk}}$ can run $\Ver(\ket{\mathrm{pk}},m,\sigma)$. The one-time security property allows the verifier to have confidence that, if the signer only used the secret key $\mathrm{sk}$ once to sign a message, then the signature is valid. 

With the definition of signature schemes and their security in hand, we now present our NISQ-friendly signature scheme. For concreteness, we base it on the random circuit OWSG from \Cref{cor:robust-owsg}; let $t,k$ be the parameters from the corollary. For simplicity we describe signatures for \emph{single bit} messages; this can be extended to many-bit messages in a straightforward way.  


\begin{longfbox}
    \begin{protocol} {\bf NISQ-friendly signature scheme with quantum public keys} \label{prot:signature} 
    \end{protocol}

    \textbf{$\mathrm{SKGen}$}. Sample $2t$ descriptions of quantum circuits $C_1^{(b)},\ldots,C_t^{(b)}$ uniformly at random, for both $b = 0,1$. Set $\mathrm{sk} = (C_i^{(b)})_{i \in [t],b \in \{0,1\}}$. 

    \vspace{4pt}

    \textbf{$\mathrm{PKGen}$}: Based on the value of $\mathrm{sk}$, output the public key state
    \[
     \ket{\mathrm{pk}} := \bigotimes_b \ket{C_1^{(b)}} \otimes \cdots \otimes \ket{C_t^{(b)}}~.
    \]

    \vspace{4pt}

    \textbf{$\mathrm{Sign}$}: To sign a bit $b$, output the classical descriptions of the the corresponding circuits. Let $\sigma = (C_1^{(b)},\ldots,C_t^{(b)})$. 

    \vspace{4pt}
    
    $\mathrm{Ver}$: Given the public key $\ket{\mathrm{pk}}$, a bit $b$, and a candidate signature $\sigma$, first interpret $\sigma$ as a tuple $(D_1,\ldots,D_t)$ where each $D_j$ is a description of a circuit. For each $j \in [t]$, apply $D_j^\dagger$ to the $n$ qubits of $\ket{C_j^{(b)}}$, and measure the $n$ qubits. If the result is all zeroes, then set $E_j = 1$, otherwise set $E_j = 0$. If $E_1 + \cdots + E_t \geq k$, then accept. Otherwise, reject.  
\end{longfbox}

\begin{lemma}[NISQ-friendly digital signatures]
\label{lem:signatures}
    Let $\cN$ denote the noise model from \Cref{cor:robust-owsg}, and assume the same $\eps$-No-Learning conjecture as in the corollary. 
    Then, the digital signature scheme described in \Cref{prot:signature} satisfies correctness and one-time security. Furthermore, the scheme it is noise-robust against $\cN$, in that the verification procedure succeeds with high probability, even on a noisy quantum computer: 
      \[
        \Pr \left [ \widetilde{\mathrm{Ver}}(\ket{\mathrm{pk}}, m, \sigma) \text{ accepts}:  \begin{array}{c} \mathrm{sk} \leftarrow \mathrm{SKGen}(1^n) \\ \ket{\mathrm{pk}} \leftarrow \widetilde{\mathrm{PKGen}}(\mathrm{sk}) \\ \sigma \leftarrow \mathrm{Sign}(\mathrm{sk}, m) \end{array} \right] = 1 - \negl(n)~.
    \]
    Here $\widetilde{\Ver}$ and $\widetilde{\mathrm{PKGen}}$ denote the noisy executions of the algorithms $\Ver$ and $\mathrm{PKGen}$ on a quantum computer with noise model $\cN$.
\end{lemma}

We don't consider noisy versions of $\mathrm{SKGen}$ and $\mathrm{Sign}$ because these are entirely classical algorithms, which we can run on a noiseless classical computer. 

\begin{proof}
The correctness property is straightforward to verify, and the security of the digital signature scheme follows from the security of the OWSG from \Cref{cor:robust-owsg} and the No-Learning Assumption. The noise robustness of the digital signature verification follows from the noise robustness of the underlying OWSG proved in \Cref{cor:robust-owsg}. 
\end{proof}
\newpage

\paragraph{A concrete NISQ-friendly proposal:} Let us consider a $2$D brickwork circuit having $n = 10$ qubits and depth $d = 10$. To recap, according to \cite{landau2024learning}, to reconstruct the state to a fidelity of $1 - \epsilon$ with $1 - \delta$ success probability, one needs,
\[
M = \frac{n^4 \cdot 2^{d^2}}{\epsilon^4} \log \frac{n}{\delta}
\]
quantum samples as input and
\[
T = \frac{n^4 \cdot 2^{d^2}}{\epsilon^4} \log \frac{n}{\delta} + \left( \frac{nkd^3}{\epsilon} \right)^{d^3}
\]
time steps. Let's fix the success probability $\delta$ to be $0.99$. Now, it can be seen by plugging in the values that even to get a meager fidelity of $0.01$ (i.e, setting $\epsilon = 0.99$), one needs at least
$\sim 10^{34}$ 
samples as well as timesteps. This shows that both the sample and time complexity are too large to be practical, even for a meager number of qubits and depth, which is heuristic evidence that the task of learning the state is hard. 

Based on this evidence, we can also estimate the public key size needed to run the secure digital signature protocol (\Cref{prot:signature}.) If we consider the same devices from Google's RCS demonstration in \cite{morvan2023phasetransitionrandomcircuit}, the fidelity for a $67$ qubit, depth $32$ circuit was $10^{-3}$. Assuming the same noise model as in \cite{morvan2023phasetransitionrandomcircuit}, the fidelity $\eta = \exp(-cn d)$, where $c$ is the noise rate per gate. Then, by plugging in the values of $n$ and $d$, we get $c \approx 3 \times 10^{-3}$. 

Furthermore, in the digital signature scheme of \Cref{prot:signature}, the public key length of Alice is given by $nt$, where $t$ is the number of rounds in the threshold repetition scheme. Moreover, Bob needs to check that at least \[ k = \left( \eta - \sqrt{\frac{n}{2t}} \right) t \]
rounds succeed. 

Now, let us fix the number of qubits to be $20$ and the depth to be $20$ units. As discussed earlier, this is far beyond the reaches of the existing learning algorithms. The fidelity, for these parameter choices, is

\[
\eta = \exp\left(-3 \times 10^{-3} \times 20 \times 20 \right) \approx 0.28.
\]

Let $t = 200$. Then, the size of Alice's public key is $4,000$ qubits. The number of rounds, that needs to be passed on Bob's end, to ensure a negligible error (exponentially suppressed in number of qubits; i.e. roughly ~$2^{-20}$) during the verification check is 
\[
\sim \left(0.28 - \sqrt{\frac{20}{400}}\right) \cdot 200 \approx 11.
\]
If we fix $n = 20$ and $k = 11$, then the following is a plot of the size of the public key required versus the depth of the circuit. Deep circuits are harder to learn, which enhances their security guarantees. However, the public key size also increases significantly, as noise causes exponential decay in signal fidelity. This is shown in Figure \ref{fig:publickey-depth}.
\begin{figure}[h!]
    \centering
    \includegraphics[width=0.8\textwidth]{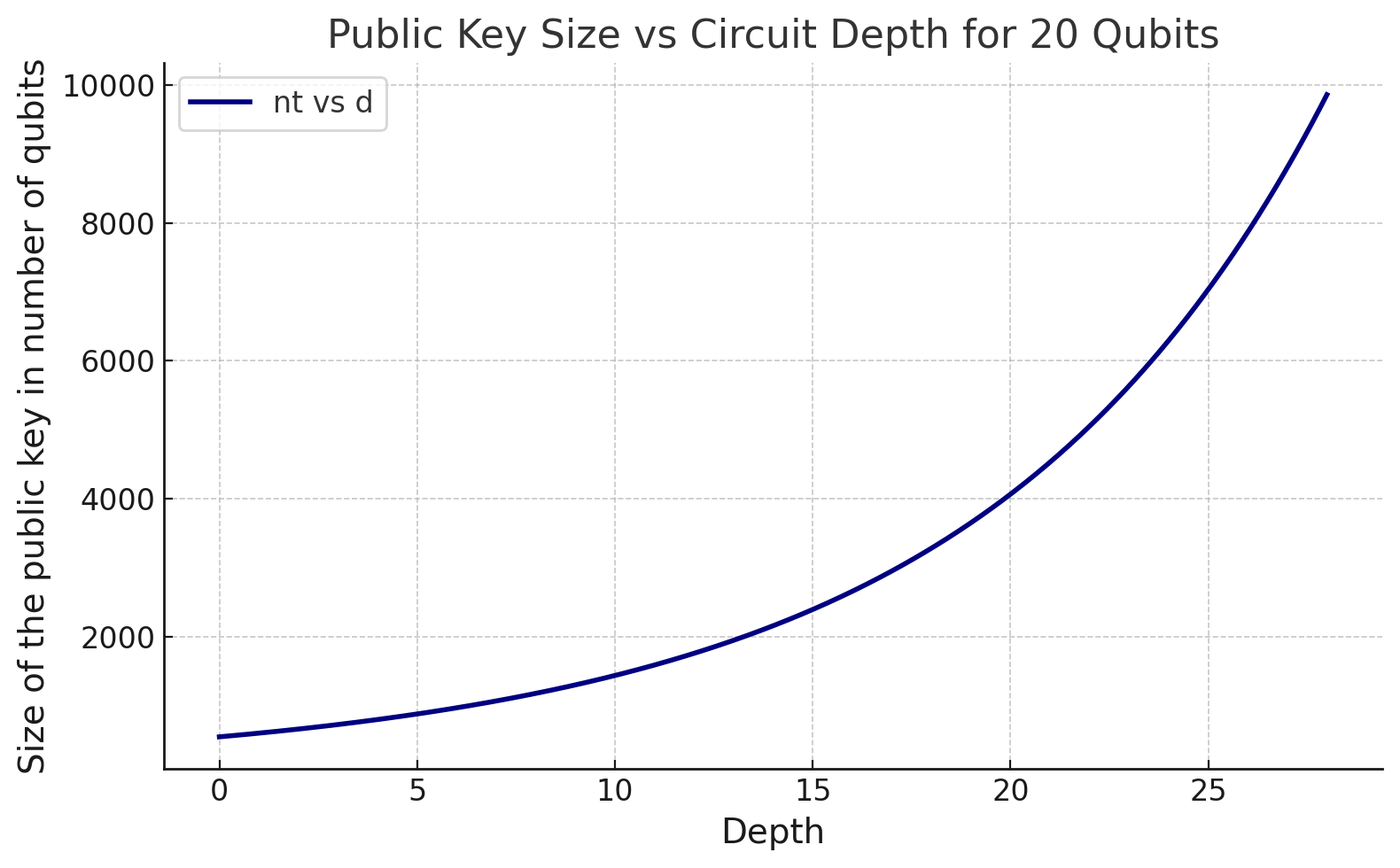}
    \caption{A plot of public key size versus circuit depth.}
    \label{fig:publickey-depth}
\end{figure}

Although our security proofs were proved in the asymptotic limit, we have attempted to give evidence that cryptographic functionalities should kick in for reasonable system sizes, for realistic systems that exist right now. However, an important caveat is that our hardness evidence for No-Learning at finite system sizes depends on just one algorithm of \cite{landau2024learning}---it is very plausible that there is a much more optimized version of these learners for more specialized ensembles, like the one Google is using in \cite{morvan2023phasetransitionrandomcircuit}, by exploiting unique properties of the same. We leave it for future work to rigorously analyze the tradeoffs between performance and runtime of existing learning algorithms. Such analysis is crucial for obtaining further evidence for the validity of our hardness assumptions.  
\newpage

\subsection*{Acknowledgements}
We thank John Bostanci, Jonas Helsen, Yihui Quek, and Abhinav Deshpande for helpful discussions. 
B.F. and S.G. acknowledge support from AFOSR
(FA9550-21-1-0008).  This material is based upon work partially
supported by the National Science Foundation under Grant CCF-2044923
(CAREER), by the U.S. Department of Energy, Office of Science,
National Quantum Information Science Research Centers (Q-NEXT) and by
the DOE QuantISED grant DE-SC0020360. H.Y. acknowledges support from AFOSR award FA9550-23-1-0363, NSF CAREER award CCF-2144219, NSF award CCF-2329939, and the Sloan Foundation. M.S. acknowledges support from the NSF award QCIS-FF: Quantum Computing \& Information Science Faculty Fellow at the University of Illinois Urbana-Champaign (NSF 1955032). This work was done in part while some of the authors were visiting the Simons Institute for the Theory of Computing, supported by NSF QLCI Grant No. 2016245.

\bibliographystyle{alpha} 
{\footnotesize
\bibliography{crypto}}

\appendix
\section{A computational Chernoff bound for one-way state generators}
\label{sec:chernoff}

In this section we prove a \emph{computational Chernoff bound} for one-way state generators. Roughly speaking, this states that if it is hard to efficiently invert a OWSG $G$ with probability more than $\gamma$, then it is hard to efficiently invert noticeably more than $\gamma$ fraction of instances of the repeated OWSG $G^t$ with non-negligible probability. The reason that this is called a ``Chernoff bound'' is because it is analogous to proving that the probability a sum of independent, bounded random variables deviates far from its mean is exponentially small. The reason this bound is ``computational'' is because it only applies to efficient algorithms.  This is also commonly known as a \emph{threshold direct product theorem} in complexity theory and cryptography~\cite{impagliazzo2010constructive}. 

This is a strengthening of the OWSG hardness amplification results of Morimae and Yamakawa~\cite{morimae2022one} which states that it is not possible to efficiently invert \emph{all} $t$ instances of the repeated OWSG $G^t$ with non-negligible probability. Morimae and Yamakawa raised the question of whether a thresholded version of their result can be proven; we answer this affirmatively.

Let $G = (\Gen,\Ver)$ be an OWSG, and let $\ket{\psi_k}$ denote the output of $G$ on input $k$ (we omit mention of the security parameter $n$ for convenience). We defined the corresponding threshold repetition, denoted by $G^{t,k}$, in \Cref{prot:OWSG-threshold}. The following lemma bounds its security error.

\chernoff*
\begin{proof}
    Throughout this proof we omit the dependence on $n$ and simply write $t = t(n), k = k(n), \gamma = \gamma(n), \xi = \xi(n)$, etc. Furthermore for notational convenience we write
    \[
        \ket{\psi_{k_1,\ldots,k_t}} := \ket{\psi_{k_1}} \otimes \cdots \otimes \ket{\psi_{k_t}}
    \]
    for a tuple of keys $(k_1,\ldots,k_t)$. 
    
    We prove this via contradiction. Suppose there was a polynomial $q = q(n)$ and a QPT algorithm $A$ such that for infinitely many $n$,
    \[
        \Pr \left [ \Ver^{t,k}\Big( A(\ket{\psi_{k_1,\ldots,k_t}}^{\otimes q}), \ket{\psi_{k_1,\ldots,k_t}} \Big) : \begin{array}{c} (k_1,\ldots,k_t) \leftarrow (\{0,1\}^{r})^t \\ 
        \ket{\psi_{k_1,\ldots,k_t}} \leftarrow \Gen^t(k_1,\ldots,k_t) \end{array} \right] \geq \frac{1}{q}
    \]
  
    We now construct a QPT algorithm $B$ that inverts the original OWSG $G$ with probability noticeably greater than $\gamma$, which is a contradiction. We closely follow the analysis of the so-called Threshold Direct Product Theorem (which is another name for a computational Chernoff bound) by Impagliazzo and Kabanets~\cite{impagliazzo2010constructive}. 
    
    This algorithm $B$ is constructed in two stages. We first construct an algorithm $A'$ that tries to solve $t/2$ copies of $G$, but it either outputs $\bot$, or with high probability inverts $\sim \gamma$ fraction of instances. Then, in the second stage we design the algorithm $B$ for a single instance of $G$ that uses the zero error algorithm $A'$ as a subroutine.

    \paragraph{First stage.} We first construct the ``zero error'' algorithm $A'$. The reason it is called ``zero error'' is because the algorithm either outputs $\bot$ or outputs (with high probability) a tuple of keys that passes verification in at least $\sim \gamma$ fraction of coordinates. Furthermore, there is a non-negligible probability of outputting a tuple of keys.

\begin{longfbox}
    \begin{protocol} {\bf The ``zero error'' algorithm $A'$} \label{prot:zero-error} 
    \end{protocol}

    \textbf{Input}: $(4q \ln 8q) \cdot (q+1)$ copies of $\ket{\psi_{k_1}} \otimes \cdots \otimes \ket{\psi_{k_{t/2}}}$.

    \vspace{4pt}
    \begin{enumerate}
        \item Sample keys $k_{t/2 + 1},\ldots,k_t \in \{0,1\}^{r(n)}$.
        \item Run the generation algorithm $\Gen$ of the OWSG $G$ on the keys to obtain $q$ copies of the output states $\ket{\psi_{k_{t/2+1}}} \otimes \cdots \otimes \ket{\psi_{k_t}}$. 
        \item Sample a random permutation $\pi$ on $[t]$. 

        \item Run the algorithm $A$ on input $\bigotimes_{i \in [t]} \ket{\psi_{k_{\pi(i)}}}^{\otimes q}$ to obtain a tuple of candidate keys $(k_{\pi(1)}',\ldots,k_{\pi(t)}')$. In other words, the copies of $\ket{\psi_{k_i}}$ states are permuted according to $\pi$ before being passed into $A$, and the output of $A$ is unpermuted according to $\pi$. 

        
        \item Run the threshold verification procedure 
        $\Ver^{t/2,\gamma t/2}$ on input $((k_{t/2+1}',\ldots,k_{t}'),\ket{\psi_{k_{t/2+1},\cdots,k_{t}}})$. If it accepts, then output $(k_1',\ldots,k_{t/2}')$ and halt. Otherwise, repeat steps 1 -- 4 for at most $4q \, \ln 8q$ times. If none of the repetitions are successful, output $\bot$.
    \end{enumerate}
\end{longfbox}

First, it is easy to verify that $A'$ runs in polynomial time. Next we argue that, conditioned on not outputting $\bot$, the algorithm $A'$ successfully inverts at least a $\gamma'$ fraction of $G$ instances with high probability. 

\begin{claim}
\label{clm:chernoff1}
Let $\gamma' = \gamma + \xi/2$.  
    Consider the following probabilistic process. Sample keys $k_1,\ldots,k_{t/2}$, and generate $(4q \ln 8q) \cdot (q+1)$ copies of the corresponding states $\ket{\psi_{k_i}}$. Run $A'$ on those copies, and obtain either $\bot$ or a tuple $(k_1',\ldots,k_{t/2}')$. The following hold, where the probabilities are over the randomness of the aforementioned process. 
    \begin{enumerate}
        \item $\Pr \left [ \Ver^{t/2,\gamma' t/2}((k_1',\ldots,k_{t/2}'),\ket{\psi_{k_1,\ldots,k_{t/2}}}) \text{ accepts} \mid A' \text{ does not output $\bot$}  \right] \geq 1 - \xi/4$.
        \item $\Pr \left [ A' 
 \text{ does not output $\bot$} \right] \geq \frac{1}{8q}$.
    \end{enumerate}
\end{claim}
\begin{proof}
    This is essentially proved in~\cite[Lemma 5.6]{impagliazzo2010constructive}; although they assumed that the $\Ver$ algorithm is classical, it can be checked that the proof goes through essentially unchanged in our setting where $\Ver$ is quantum and receives a quantum input. 
\end{proof}

\paragraph{Second stage.} For the second stage we construct an algorithm $B$ that tries to solve a single instance of $G$ by embedding it into the threshold repetition $G^{t,k}$. The algorithm $B$ uses $A'$ as a subroutine.

\begin{longfbox}
    \begin{protocol} {\bf The algorithm $B$ to solve a single instance of $G$} \label{prot:zero-error} 
    \end{protocol}

    \textbf{Input}: $(128 q \ln \frac{120}{\xi}) \cdot (4q \ln 8q) \cdot (q+1)$ copies of $\ket{\psi_{k}}$.

    \vspace{4pt}
    \begin{enumerate}
        \item Sample a random index $i^* \in [t/2]$.
        \item Sample random keys $k_j$ for $j \in [t/2] \setminus \{i^*\}$ and generate $(4q \ln 8q) \cdot (q+1)$ copies of the corresponding states $\ket{\psi_{k_j}}$. 
        \item Run the algorithm $A'$ on input $(4q \ln 8q) \cdot (q+1)$ copies of $\Big ( (\bigotimes_{j < i^*} \ket{\psi_{k_j}} \Big) \otimes \ket{\psi_k} \otimes \Big ( (\bigotimes_{j > i^*} \ket{\psi_{k_j}} \Big) $.

        \item If $A'$ outputs $\bot$, then try steps 1 -- 3 again for at most $128q \, \ln \frac{120}{\xi}$ times. If none of the repetitions succeed, then output $\bot$.  Otherwise, if $A'$ outputs $(k_1',\ldots,k_{t/2}')$, then output $k_{i^*}'$.

    \end{enumerate}
\end{longfbox}

Again, it should be clear from construction that $B$ runs in polynomial time. The following argues that $B$ successfully inverts the OWSG $G$, using a polynomial number of copies of the output of $G$. 

\begin{claim}
\label{clm:chernoff2}
Let $s = (128 q \ln \frac{120}{\xi}) \cdot (4q \ln 8q) \cdot (q+1)$. Then 
\[
\Pr \left [ \Ver(B(\ket{\psi_k}^{\otimes s}), \ket{\psi_k}) \text{ accepts} : \begin{array}{c} k \leftarrow \{0,1\}^{r(n)} \\ \ket{\psi_k} \leftarrow \Gen(k) \end{array} \right] \geq \gamma + \xi/20~.
\]
\end{claim}
\begin{proof}
        This is essentially proved in~\cite[Lemma 5.8]{impagliazzo2010constructive}. Again, they prove their result for classical algorithms, but the proof goes through in our quantum setting. 
\end{proof}

Putting everything together, we get that $B$ inverts the OWSG $G$ with probability at least $\gamma + \xi/20$, which contradicts the security guarantee of $G$. Therefore the threshold repetition of $G^{t,k}$ has negligible security error. 

\end{proof}

\end{document}